\documentclass[letterpaper,11 pt, onecolumn]{IEEEtran}  

\IEEEoverridecommandlockouts                              

\usepackage{xcolor}
\usepackage{float}
\usepackage{amsmath}

\usepackage{amsthm,amssymb,graphicx,url}

\usepackage{enumitem}
\usepackage{booktabs}
\usepackage{subcaption}
 
\usepackage{appendix}

\usepackage{wrapfig}
\usepackage{soul}  
\usepackage{bm}
\usepackage{multicol}
\usepackage{mathtools}
\usepackage{gensymb}
\usepackage{comment}

\usepackage{algorithm,algorithmicx} 

\usepackage{commath}

\newlength\myindent
\DeclareMathOperator*{\argmin}{argmin}

\def\blue#1{{\color{black}#1}}

\newcommand{\infnorm}[1]{\left\lVert#1\right\rVert_\infty}

\newcommand{\lonenorm}[1]{\left\lVert#1\right\rVert_{\mathcal{L}_1}}

\def\laplace#1{\mathfrak{L}[#1]}

\def\lonew{${\mathcal{L}_1}$ }
\def\linf{{\mathcal{L}_\infty}}

\def\laplace#1{\mathfrak{L}[#1]}

\def\laplace#1{\mathfrak{L}[#1]}

\def\tilx{\tilde{x}}

\def\mbR{\mathbb{R}}

\def \mbZ{\mathbb{Z}}

\def\hsigma{\hat{\sigma}}

\def\mcC{\mathcal{C}}
\def\mcD{\mathcal{D}}

\def\mcX{\mathcal{X}}

\def\mcU{\mathcal{U}}
\def\mcX{\mathcal{X}}

\def \hatx{\hat{x}}

\def \xstar {x^\star}
\def \ustar {u^\star}

\def\trieq{\triangleq}

\def\tild{\tilde{d}}

\def\mbZ{\mathbb{Z}}

\def\mcC{{\mathcal{C}}}

\def \bbracket#1{\bm{[}#1\bm{]}}

\def \bardeltade {\bar\delta_\textup{de}}
\def \bardeltatild{\bar \delta_{\tilde d}}

\makeatletter
\def\@listiii{\leftmargin\leftmarginiii
              \labelwidth\leftmarginiii
              \advance\labelwidth-\labelsep
              \topsep\z@
              \parsep\z@
              \partopsep\z@
              \itemsep\topsep}
\makeatother
\usepackage{cleveref}

\crefname{equation}{}{} 
\crefname{lemma}{Lemma}{Lemmas}
\crefname{theorem}{Theorem}{Theorems}
\crefname{table}{Table}{Tables}
\crefname{figure}{Fig.}{Figs.}
\crefname{remark}{Remark}{Remarks}
\crefname{assumption}{Assumption}{Assumptions}
\crefname{section}{Section}{Sections}
\crefname{definition}{Definition}{Definitions}
\crefname{algorithm}{Algorithm}{Algorithms}
\crefname{proposition}{Proposition}{Propositions}
\crefname{appendix}{Appendix}{Appendices}

\usepackage{amsthm}
\newtheorem{theorem}{Theorem}
\newtheorem{lemma}{Lemma}

\theoremstyle{definition}  
\theoremstyle{definition} 
\newtheorem{assumption}{Assumption}
\theoremstyle{remark}  
\newtheorem{remark}{Remark}

\title{\LARGE \bf
Guaranteed Trajectory Tracking under Learned
Dynamics with Contraction Metrics and Disturbance Estimation}

\author{Pan Zhao, Ziyao Guo, Yikun Cheng,  Aditya Gahlawat, Hyungsoo Kang and Naira Hovakimyan 
\thanks{This work is supported in part by NASA through the ULI grant \#80NSSC22M0070,
and in part by NSF under the RI grant \#2133656}
\thanks{P. Zhao is with the Department of Department of Aerospace Engineering and Mechanics, University of Alabama, Tuscaloosa, AL 35487, USA. {\tt \small pan.zhao@ua.edu}}
\thanks{Z. Guo, Y. Cheng,  A. Gahlawat, H. Kang and N. Hovakimyan are with the Department of 
Mechanical Science and Engineering, University of
Illinois at Urbana-Champaign, Urbana, IL 61801, USA. {\tt\small ziyaog2, yikun2, gahlawat, hk15, nhovakim@illinois.edu}}%
}
\begin{document}

\maketitle
\thispagestyle{empty}
\pagestyle{empty}

\begin{abstract}
This paper presents a contraction-based learning control architecture that allows for using model learning tools to learn matched model uncertainties while guaranteeing trajectory tracking performance during the learning transients. The architecture relies on a disturbance estimator to estimate the pointwise value of the uncertainty, i.e., the discrepancy between a nominal model and the true dynamics, with pre-computable estimation error bounds, and a robust Riemannian energy condition for computing the control signal. Under certain conditions, the controller guarantees exponential trajectory convergence during the learning transients, while learning can improve robustness and facilitate better trajectory planning. 
Simulation results validate the efficacy of the proposed control architecture. 
\end{abstract}

\begin{IEEEkeywords}
Robust control;  decision-making under uncertainty; machine learning for control; robot safety
\end{IEEEkeywords}

\section{Introduction}\label{sec:introduction}

Robotic and autonomous systems often exhibit nonlinear dynamics and operate in uncertain and disturbed environments. Planning and executing a trajectory is one of the most common ways for an autonomous system to achieve a mission. However, the presence of uncertainties and disturbances, together with the nonlinear dynamics, brings significant challenges to safe planning and execution of a trajectory. Built upon contraction theory and disturbance estimation, this paper presents a trajectory-centric learning control approach that allows for using different model learning tools to learn uncertain dynamics while providing guaranteed tracking performance in the form of exponential trajectory convergence throughout the learning phase. Our approach hinges on control contraction metrics (CCMs) and uncertainty estimation. 

\subsection{Related Work}

\blue{{\it Control design methods for uncertain systems} can be roughly classified into adaptive/robust approaches and learning-based approaches.   Robust approaches such as $H_\infty$ control \cite{Zhou98essentials}, sliding mode control \cite{edwards1998sliding-book}, and robust/tube model predictive control (MPC) \cite{mayne2005robust-tube-mpc,mayne2014mpc-survey}, usually consider parametric uncertainties or bounded disturbances,  and design controllers to achieve certain performance despite the presence of such uncertainties. Disturbance-observer (DOB) based control and related methods such as active disturbance rejection control  (ADRC) \cite{han2009adrc} usually lump parametric uncertainties, unmodeled dynamics, and external disturbances together as a ``total disturbance'', estimate it via an observer such as DOB and extended state observer (ESO) \cite{han2009adrc}, and then compute control actions to compensate for the estimated disturbance \cite{chen2015dobc}. On the other hand, adaptive control methods such as model reference adaptive control (MRAC) \cite{ioannou2012robust} and  \lonew adaptive control \cite{naira2010l1book} 
rely on online adaptation to estimate parametric or non-parametric uncertainties and use of the estimated value in the control design to provide stability and performance guarantees.
While significant progress has been made in the linear setting, {\it trajectory tracking for nonlinear uncertain systems} with {\it transient performance guarantees} has been less successful in terms of analytical quantification, yet it is required for safety guarantees of robotic and autonomous systems. }

{\it Contraction theory}~\cite{lohmiller1998contraction} provides a powerful tool for analyzing general nonlinear systems in a differential framework and is focused on studying the convergence between pairs of state trajectories towards each other, i.e., incremental stability. It has recently been extended for constructive control design, e.g., via control contraction metrics (CCMs) \cite{manchester2017control,tsukamoto2020robust-stochastic}.
Compared to incremental Lyapunov function approaches for studying incremental stability, contraction metrics present an {\it intrinsic} characterization of incremental stability (i.e., invariant under change of coordinates); additionally, the search for a CCM can be achieved using the sum of squares (SOS) optimization or semidefinite programming \cite{singh2019robust}, and DNN optimization \cite{tsukamoto2020neural-contraction,sun2020learning-ccm}. 
For nonlinear uncertain systems, CCM has been integrated with adaptive and robust control to address parametric \cite{lopez2020adaptive-ccm} and non-parametric uncertainties \cite{zhao2022robust-ccm-de,lakshmanan2020safe}. The issue of bounded disturbances in contraction-based control has been tackled through robust analysis \cite{singh2019robust} or synthesis \cite{zhao2022tube-rccm-ral,manchester2018rccm}. For more work related to contraction theory and CCM for nonlinear stability analysis and control synthesis, the readers can refer to a recent tutorial paper \cite{tsukamoto2021contraction-tutorial} and the references therein. 

On the other hand, recent years have witnessed an increased use of {\it machine learning (ML) to learn dynamics models}, which are then incorporated into control-theoretic approaches to generate the control law. 
{For model-based learning control with {\it safety  and/or transient performance guarantees}, most of the existing research relies on {\it} quantifying the learned model error, and {\it robustly} handling such an error in the controller design or analyzing its effect on the control performance
\cite{khojasteh2020probabilistic-abbrev,berkenkamp2015safe,chou2021model-error}. 
As a result, researchers have largely relied on Gaussian process regression (GPR) to learn uncertain dynamics, due to its inherent ability to quantify the learned model error \cite{khojasteh2020probabilistic-abbrev,berkenkamp2015safe}. 
Additionally, in almost all the existing research, the control performance is directly determined by the quality of the learned model, i.e., a poorly learned model naturally leads to poor control performance. 
   Deep neural networks (DNNs) were used to approximate state-dependent uncertainties in adaptive control design in \cite{joshi2019deep-mrac,sun2021lyapunov-dnn}. 
However, these results only provide asymptotic (i.e., no transient) performance guarantees at most, and investigate pure control problems without considering planning. Moreover, they either consider linear nominal systems or leverage feedback linearization to cancel the (estimated) nonlinear dynamics, which can only be done for fully actuated systems. In contrast, this paper considers the planning-control pipeline and does not try to cancel the nonlinearity, thereby allowing the systems to be underactuated. 
In \cite{shi2019neurallander,chou2021model-error}, the authors used DNNs for batch-wise learning of uncertain dynamics from scratch; 
however, good tracking performance cannot be achieved when the learned uncertainty model is poor. 
\begin{figure}
    \centering
    \includegraphics[width=0.9\columnwidth]{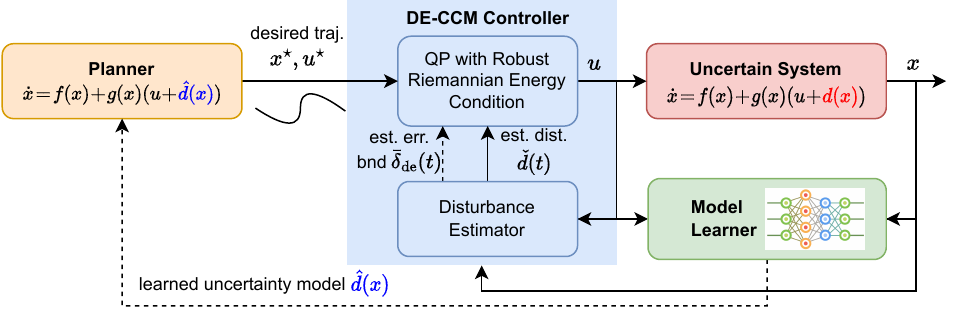}
    \vspace{-2mm}
    \caption{Proposed control architecture incorporating learned dynamics}
    \label{fig:ccm-dnn-framework}
    \vspace{-2mm}
\end{figure}

{\it Statement of Contributions}: We propose a contraction-based learning control architecture for nonlinear systems with matched uncertainties (depicted in \cref{fig:ccm-dnn-framework}). The proposed architecture allows for using different  ML tools, e.g., DNNs, for learning the uncertainties while guaranteeing exponential trajectory convergence under certain conditions throughout the learning phase. It leverages a disturbance estimator with a pre-computable estimation error bound (EEB) and a robust Riemann energy condition to compute the control signal. It is empirically shown that learning can help improve the robustness of the controller and facilitate better trajectory planning. We demonstrate the efficacy of the proposed approach using a planar quadrotor example. 

\blue{This work builds on \cite{zhao2022robust-ccm-de} with several key differences. The authors of \cite{zhao2022robust-ccm-de} introduce a robust tracking controller utilizing CCM and disturbance estimation without involving model learning. In contrast, this work adapts the controller to handle scenarios where machine learning tools are used to learn unknown dynamics, offering tracking performance guarantees throughout the learning phase. Additionally, this research empirically showcases the advantages of integrating learning to improve trajectory planning and strengthen the robustness of the closed-loop system, aspects that are not explored in \cite{zhao2022robust-ccm-de}.}


{\it Notations}. Let $\mathbb{R}^n$ and $\mathbb{R}^{m\times n}$ denote  the $n$-dimensional real vector space and the space of real $m$ by $n$ matrices, respectively.  ${0}$ and $I$ denote a zero matrix and an identity matrix of compatible dimensions, respectively.
$\infnorm{\cdot}$ and $\norm{\cdot}$ denotes the $\infty$-norm and $2$-norm of a vector/matrix, respectively. 
Given a vector $y$,  let $y_i$ denote its $i$th element. For a vector $y\in\mbR^n$ and a matrix-valued function $M: \mbR^n\rightarrow \mbR ^{n\times n}$, let $\partial_y M(x)\trieq \sum_{i=1}^n\frac{\partial M(x)}{\partial x_i}y_i$ denote the directional derivative of $M(x)$ along $y$. 
For symmetric matrices $P$ and $Q$, $P>Q$ ($P\geq Q$) means $P-Q$ is positive definite (semidefinite). $\langle X\rangle$ stands for $X +X^\top$. Finally, we use $\ominus$ to denote the Minkowski set difference.

\section{Preliminaries and Problem Setting}
Consider a nonlinear control-affine system 
\begin{equation}\label{eq:dynamics}
    \dot{x}(t) = f(x(t))+B(x(t))(u(t) + d(x(t))),\quad x(0)=x_0,
\end{equation}
where $x(t)\in \mcX \subset \mbR^n$ and $u(t) \in \mcU \subset \mbR^m$ are state and input vectors, respectively, $f:\mbR^n\rightarrow \mbR^n$ and $B:\mbR^n\rightarrow \mbR^m$ are known  functions that are locally Lipschitz, $d:\mbR^n\rightarrow \mbR^m$ is an unknown function denoting the {\it matched} model uncertainties. $B(x)$ is assumed to have full column rank for all $x\in \mcX$.
Additionally, $\mcX$ is a compact set that includes the origin, and  $\mcU$ is the control constraint set  defined as $\mcU\trieq \{u\in \mbR^m: \underline{u}\leq u \leq \bar{u}_i, \ i=1,\dots,m\}$, where $\underline{u}_i$ and $\bar{u}_i$ represent the lower and upper bounds of the $i$th control channel, respectively. 
\begin{remark}
The matched uncertainty assumption is common in adaptive control \cite{ioannou2012robust} or disturbance observer-based control \cite{chen2015dobc}, and is made in existing related works 
\cite{lakshmanan2020safe,zhao2022robust-ccm-de,lopez2020adaptive-ccm}.
\end{remark}
\begin{assumption}\label{assump:lipschitz-bound-d-B}
There exist known positive constants $L_B$, $L_d$ and $b_d $ such that for any $x,y \in \mcX$, the following holds:
\begin{IEEEeqnarray*}{rl}
\left\| {B(x) - B(y)} \right\| &\le {L_B}\left\| {x - y} \right\|,\\
\left\| {d(x) - d(y)} \right\| &\le {L_d}\left\| {x - y} \right\|,\  \left\| {d(x)} \right\|  \le {b_d}. \label{eq:lipschitz-cond-d-B}
\end{IEEEeqnarray*}
\end{assumption}
\begin{remark} Assumption~\ref{assump:lipschitz-bound-d-B}  {\it does not assume} that the system states stay in $\mcX$ (and thus are bounded). We will prove the boundedness of $x$ later in \cref{them:DE-CCM}.
 Assumption~\ref{assump:lipschitz-bound-d-B} merely indicates that  $d(x)$ is {\it locally Lipschitz} with a known {\it bound}  on the Lipschitz constant and is bounded by a prior known constant {\it in the compact set $\mcX$}. 
\end{remark}
 Assumption~\ref{assump:lipschitz-bound-d-B} is not very restrictive as the local Lipschitz bound in $\mcX$ for $d(x)$  can be conservatively estimated from prior knowledge. Additionally, given the local Lipschitz constant bound $L_d$ and the compact set $\mcX$, we can always derive a uniform bound using Lipschitz property if a bound for $d(x)$ for any $x$ in $\mcX$ is known. For example, supposing $\norm{d(0)}\leq b_d^0$,  we have $\norm{d(x)}\leq b_d^0 + L_d\max_{x\in\mcX} \norm{x}$. In practice, leveraging prior knowledge about the system can result in a tighter bound than the one based on Lipschitz continuity. Thus, we assume a uniform bound.
Under Assumption~\ref{assump:lipschitz-bound-d-B}, it will be shown later (in \cref{sec:sub-disturbance-estimation}) that the pointwise value of $d(x(t))$ at any time $t$ can be estimated with pre-computable estimation error bounds (EEBs). 

\subsection{Learning Uncertain Dynamics}
Given a collection of data points $\{(x_i,d_i)\}_{i=1}^N$ with $N$ denoting the number of data points, the uncertain function $d(x)$ can be learned using ML tools. As a demonstration purpose, we choose to use DNNs, due to their significant potential in dynamics learning attributed to their expressive power and the fact that they have been rarely explored for dynamics learning with safety and/or performance guarantees. 
Denoting the learned function as $\hat d(x)$ and the model error as $\tilde d(x) \trieq d(x) - \hat d(x)$, the actual dynamics \eqref{eq:dynamics} can be rewritten as 
\begin{equation}\label{eq:dynamics-w-dnn}
  \hspace{-2mm}      \dot{x}  \!=\! f(x)\!+\!B(x)(u \!+\!\hat d(x) \!+ \!\tilde d(x)) \!=\!    F_l(x,u) \!+\! B(x)\tilde d(x),
\end{equation}
where 
\begin{equation}\label{eq:Fl-defn}
    F_l(x,u) \trieq \overbrace{f(x)+B(x)\hat d(x)}^{ {\textstyle \trieq \hat f(x)}} + B(x)u.
\end{equation}
The learned dynamics can now be represented as 
\begin{equation}\label{eq:dynamics-learned}
{\dot x}^\star = F_l(\xstar,\ustar).
\end{equation}
\begin{remark}
The above setting includes the special case of {\em no learning}, corresponding to $\hat d(x) \trieq 0$. 
\end{remark}
  Note that the performance guarantees provided by the proposed framework are {\it agnostic to the model learning tools used}, as long as the following assumption can be satisfied. 
\begin{assumption}\label{assump:uniform-err-bnd}
We are able to obtain a uniform error bound for the learned function $\hat d(x)$, i.e., we could compute a constant $\bardeltatild $ such that 
\begin{equation}\label{eq:uniform-err-bnd}
   \max _{x\in\mcX} \norm{\tild (x)}\leq \bardeltatild.  
\end{equation}
\end{assumption}
\begin{remark}\label{remark:uniform-err-bnd}
    Assumption~\ref{assump:uniform-err-bnd} can be easily satisfied when using a broad class of ML tools. For instance, when Gaussian processes are used,  a uniform error bound (UEB) can be computed using the approach in \cite{lederer2019uniform}. When using DNNs, we could use spectral-normalized DNNs (SN-DNNs) \cite{miyato2018spectral-normalization-abbrev} (to enfoce that $\hat d(x)$ has a Lispitchz bound $L_d$ in $\mcX$) and compute the UEB as follows. 
    Obviously, since both $d(x)$ and $\hat d(x)$  have a local Lipschitz bound $L_d$ in $\mcX$, the model error $\tilde d(x)$ has a local Lipschitz bound $2L_d$ in $\mcX$. As a result, given any point $x^\ast\in\mcX$, we have 
$$
    \|\tilde d(x^\ast)\| \leq \min\limits_{i\in \{1,\dots,N\}} \left( \|\tilde d(x_i)\|+ 2L_d \norm{x^\ast-x_i}\right), 
$$ where $x_i$ is one of the $N$ number of data points. The preceding inequality implies \cref{eq:uniform-err-bnd} holds with $\bar \delta_{\tild} = \max_{x^\ast \in\mcX}\min_{i\in \{1,\dots,N\}} \left( \|{\tilde d(x_i)}\|+ 2L_d \norm{x^\ast-x_i}\right)$.
\end{remark}
\subsection{Problem Setting}
The learned dynamics \cref{eq:dynamics-learned} (including the special case of $\hat d= 0$) can be incorporated in a motion planner or trajectory optimizer to plan a desired trajectory $(\xstar(\cdot),\ustar(\cdot))$ to minimize a specific cost function. Suppose Assumptions~\ref{assump:lipschitz-bound-d-B}, \ref{assump:uniform-err-bnd}, and \ref{assump:ccm-exists-for-nominal-sys} hold. The focus of the paper includes 
\begin{enumerate}[label =(\roman*)]
    \item designing a feedback controller to track the desired state trajectory $\xstar(\cdot)$ with {\it guaranteed tracking performance} despite the presence of the model error $\tilde d(x)$, and
    \item empirically demonstrating the benefits of learning in improving the robustness and reducing the cost associated with the actual trajectory. 
\end{enumerate}

In the following, we will present some preliminaries on CCM and uncertainty estimation used to build our solution. 

\subsection{CCM for the Nominal Dynamics}\label{sec:sub-ccm-review}
CCM extends contraction analysis \cite{lohmiller1998contraction} to controlled dynamic systems, where the analysis simultaneously seeks a controller and a metric that characterizes the contraction properties of the closed-loop system  \cite{manchester2017control}. According to \cite{manchester2017control}, a symmetric matrix-valued function $M(x)$ serves as a strong CCM for the {nominal} ({\it uncertainty-free}) system
\begin{equation}\label{eq:dynamics-nominal}
    \dot x = f(x) + B(x) u,
\end{equation}
in $\mcX$, if there exist positive constants $\alpha_1$, $\alpha_2$ and $\lambda$ such that
\begin{subequations}\label{eq:ccm-condition-strong}
\begin{gather}
 \alpha_1 I \!\leq\!  M(x)\!\leq\! \alpha_2 I \label{eq:ccm-uniform-bnd}\\
    \delta_x ^\top MB=0 \! \Rightarrow \!  \delta_x^\top\!\left(\!\left \langle \!M{\textstyle\frac{\partial f}{\partial x}}\!\right\rangle \!+ \!\partial_f M\!+\!2\lambda M\!\right) \!\delta_x \!\leq\! 0  \label{eq:ccm-condition-strong-1}\\
\!\!\left\langle \!M{\textstyle\frac{\partial b_i}{\partial x}}\!\right\rangle \!+\! \partial_{b_i}M\! =\! 0, \ \!  i=1,\cdots,m, \label{eq:ccm-condition-strong-2}
\end{gather}
\end{subequations}
hold for all $\delta_x \neq 0$ and $x\in\mcX$. 

\begin{assumption}\label{assump:ccm-exists-for-nominal-sys}
    There exists a strong CCM $M(x)$ for the nominal system \cref{eq:dynamics-nominal} in $\mcX$, i.e., $M(x)$ satisfies \cref{eq:ccm-uniform-bnd,eq:ccm-condition-strong-1,eq:ccm-condition-strong-2}.    
\end{assumption}
\begin{remark}
Similar to the synthesis of Lyapunov functions, given dynamics, a strong CCM can be systematically synthesized using convex optimization, more specially, sum of squares programming \cite{manchester2017control,singh2019robust,zhao2022tube-rccm-ral}.  
\end{remark}
Given a CCM $M(x)$, a feasible trajectory $(\xstar(\cdot), \ustar(\cdot))$ satisfying the nominal dynamics \cref{eq:dynamics-nominal}, and the actual state $x(t)$ at $t$, 
the control signal can be constructed as follows \cite{manchester2017control, singh2019robust}. At any $t>0$, compute a minimal-energy path (termed as {\it geodesic}) $\gamma(\cdot,t)$ connecting  $x(t)$ and $\xstar(t)$, e.g., using the pseudospectral method \cite{leung2017pseudospectral-geodesic}. Note that the geodesic is always a straight line segment if the metric is constant.  Next, compute the Riemann energy of the geodesic, defined as $E(\xstar(t),x(t))= \int_0^1 \gamma_s(s,t)^\top M(\gamma(s,t)))\gamma_s(s,t)ds$, where $\gamma_s(s)\trieq \frac{\partial \gamma}{\partial s}$. Finally, by interpreting the Riemann energy as an incremental control Lyapunov function, we can construct a control signal $u(t)$ such that 
\begin{align}
    \dot E(\xstar(t),x(t), u(t)) & \leq -2 \lambda E(\xstar(t),x(t)),\label{eq:Edot-E-ineq} 
\end{align} 
where $\dot E(\xstar,x,u)  = 2\gamma_s^\top(1,t)M(x) \dot x - 2\gamma_s^\top(0,t)M(\xstar)\dot x^\star$ with the dependence on $t$ omitted for brevity, 
$\dot x(t)$ is defined in \cref{eq:dynamics-nominal} and  $\dot x^\star (t) = f(\xstar(t))+B(\xstar(t))\ustar(t)$. 
 In practice, one may want to compute $u(t)$ with a minimal $ \norm{u(t)-\ustar(t)}$ such that \cref{eq:Edot-E-ineq} holds, which can be achieved by setting $u(t) =\ustar(t) + k(t,\xstar,x) $ with $k(t,\xstar,x) $  obtained via solving a quadratic programming (QP) problem \cite{manchester2017control,singh2019robust}:
\begin{align}
  k(t,\xstar,x) =\argmin_{k\in\mbR^m}  & \norm{k-\ustar(t)}^2 \label{eq:qp-nominal}  \\
 \hspace{-5mm} \textup{s.t. }   
  \dot E(\xstar(t),x(t), \ustar(t)+k) & \leq -2 \lambda E(\xstar(t),x(t)) 
  \label{eq:Edot-E-ineq-detailed}
\end{align}
at each time $t$. 
The problem \cref{eq:qp-nominal} is commonly referred to as the pointwise minimum-norm control problem and possesses an analytic solution  \cite{freeman2008robust-clf-book}.
The performance guarantees provided by the CCM-based controller can be summarized in the following theorem.  
\begin{lemma} \label{them:synthesis-nominal}
\cite{manchester2017control} Suppose Assumption~\ref{assump:ccm-exists-for-nominal-sys} holds for the nominal system \cref{eq:dynamics-nominal} with positive constants $\alpha_1$, $\alpha_2$ and $\lambda$.  Then, a control law satisfying \cref{eq:Edot-E-ineq} universally exponentially stabilizes the system \cref{eq:dynamics-nominal}, which can be expressed mathematically as
\begin{equation}\label{eq:ues-math}
    \norm{x(t)-\xstar(t)} \leq \sqrt{\frac{\alpha_2}{\alpha_1}} \norm{x(0)-\xstar(0)} e^{-\lambda t}.
\end{equation}
\end{lemma}

The following lemma from \cite{singh2019robust} establishes a bound on the tracking control effort from solving \cref{eq:qp-nominal}. 
{\begin{lemma}\label{lemma:u-bnd-nominal}
\cite[Theorem 5.2]{singh2019robust} 
For all $\xstar,~x\in \mcX$ such that $\mathtt d(\xstar,x)\leq \bar {\mathtt d}$ for a scalar constant $\bar {\mathtt d}\geq 0$, the control effort from solving \cref{eq:qp-nominal} is bounded by 
\begin{equation}\label{eq:k-bound-nominal-w-bardeltau-defn}
    \norm{k^\ast(t,\xstar,x)}\leq \bar \delta_u \trieq \frac{\bar {\mathtt d}}{2}\sup_{x\in\mcX}\frac{\bar \lambda(L^{T}GL)}{\underline \sigma_{> 0}(B^T L^{-1})},
\end{equation}
where $G(x) \trieq\left \langle \!M\frac{\partial f}{\partial x}\!\right\rangle + \partial_f M +2\lambda M\!$, $L(x)$ satisfies $M(x) = L(x)L^T(x)$, $\bar\lambda(\cdot)$ represents the largest eigenvalue, and $\underline \sigma_{> 0}(\cdot)$ denotes the smallest non-zero singular value. 
\end{lemma}}

\subsection{Uncertainty Estimation with Computable Error Bounds}\label{sec:sub-disturbance-estimation}
We leverage a disturbance estimator described in \cite{zhao2022robust-ccm-de} 
to estimate the value of the uncertainty $d(x(t))$ at each time instant. More importantly, an estimation error bound can be pre-computed and systematically reduced by tuning a parameter of the estimator. 
The estimator comprises a state predictor and an update law. The state predictor is defined by
{
\begin{equation}\label{eq:state-predictor}
    \dot{\hatx}(t)  =  f(x(t)) +B(x(t)) u (t)+ \hsigma(t) -a\tilx(t),\ \hatx(0)= x_0 ,
\end{equation}}where $\tilx(t) \triangleq \hatx(t) - x(t)$ denotes the prediction error,  $a>0$ is a  scalar constant. The estimate, $\hsigma(t)$, is given by a piecewise-constant update law: 
\begin{equation}\label{eq:estimation_law}
\left\{
\begin{split}
   \hsigma(t)
    &= 
    \hsigma(iT),  \quad t\in [iT, (i+1)T),\\
    \hsigma(iT)
    &= - {{a}/{(e^{aT}-1})}\tilde{x}(iT),  
\end{split}\right.
\end{equation}
where $T$ is an estimation sampling time, and $i=0,1,2,\cdots$. Finally, the  value of $d(x)$ at time $t$ is estimated as
\begin{equation}\label{eq:dcheck-hsigma-relation}
   \check {d}^0(t) = B^\dagger(x(t))\hsigma(t),
\end{equation}
where $B^\dagger(x(t))$ is the pseudoinverse of $B(x(t))$. 
The following lemma
establishes the estimation error bound associated with the disturbance estimator defined by \eqref{eq:state-predictor} and \eqref{eq:estimation_law}. \cref{lemma:estiamte-error-bound} can be considered as a simplified version of \cite[Lemma~4]{zhao2022robust-ccm-de} in the sense that the uncertainty considered in \cite{zhao2022robust-ccm-de} can explicitly depend on both time and states, i.e., represented as $d(t,x)$, while the uncertainty in this paper depends on states only.  
The proof is similar to that in \cite{zhao2022robust-ccm-de} and is omitted for brevity. 
\begin{lemma} \cite{zhao2022robust-ccm-de}  \label{lemma:estiamte-error-bound}
Given the dynamics \eqref{eq:dynamics} subject to Assumption~\ref{assump:lipschitz-bound-d-B}, and the disturbance estimator in \eqref{eq:state-predictor} and \eqref{eq:estimation_law}, for $\xi \geq T$, if
\begin{equation}\label{eq:xu-in-XU-0-tau}
x(t)\in\mcX,\ \forall t \in[0,\xi],\  u(t)\in \mcU,\ \forall t \in [0,\xi),
\end{equation}
the estimation error can be bounded as 
{
\begin{align}
    \hspace{-3mm}  \norm{\check d^0\!(t)\!-\!d(x(t))} & \!\leq\!
    \bardeltade^0\!(t) \!\trieq \!\! \left\{
    \begin{array}{ll}
  \!\!\! b_d, \ \forall 0\leq t<T, \\
  \!\!\! {\textstyle \alpha(T)  \max\limits_{x\in \mcX} \norm{B^\dagger(x)}}, \ \! \forall t \!\geq\! T\!,\hspace{-5mm}
    \end{array}
    \right.\label{eq:estimation_error_bound}
\end{align}}for all $t$ in $[0,\xi]$, where 
\begin{align}
    \alpha(T) & \triangleq   2 \sqrt{n} \phi T( {L_d}{\max\limits _{x \in \mathcal{X}}}\left\| {B(x)} \right\| \!+\! {L_B}b_d) +\!(1\!-\!e^{-aT})\sqrt n b_d \mathop{\max}\limits_{x\in \mcX}\norm{B(x)} \\
    \phi & \trieq  \max\limits_{x\in \mcX, u\in \mcU}\norm{f(x) + B(x)u}+b_d\max\limits_{x\in \mcX}\norm{B(x)},  \label{eq:phi-defn}
\end{align}
with constants $L_B$, $L_d$ and $b_d$~from   Assumption~\ref{assump:lipschitz-bound-d-B}.
Moreover, 
$
    \lim_{T\rightarrow 0}  \bardeltade^0(t)  = 0,
$ for any $t\geq T$.
\end{lemma}

\begin{remark}
According to Lemma~\ref{lemma:estiamte-error-bound}, the estimation error after a single sampling interval can be arbitrarily reduced by decreasing $T$. 
\end{remark}

\begin{remark}\label{remark:est_err_bnd_for_practical_implementation}
As explained in \cite[Remark 5]{zhao2022robust-ccm-de}, the error bound $\bar \delta_\textrm{de}^0(t)$ can be quite conservative primarily due to the conservatives introduced in computing $\phi$ and $\alpha(T)$. For {\it practical implementation}, one could benefit from empirical studies, such as conducting simulations with a selection of user-defined functions of $d(x)$ to identify a more refined bound than $\bar \delta_\textrm{de}^0(t)$ defined in \cref{eq:estimation_error_bound}. 
\end{remark}
\section{Disturbance Estimation-Based Contraction Control Under Learned Dynamics}\label{sec:robust-CCM}
This section introduces an approach based on CCM and disturbance estimation to ensure the UES of the uncertain system \cref{eq:dynamics-w-dnn} even when the learned model $\hat d(x)$ is poor. 

\subsection{CCMs and Feasible Trajectories for the True System}

The first step is to search for a valid CCM for the uncertain system \cref{eq:dynamics-w-dnn}. Due to the special structure of \cref{eq:dynamics-w-dnn} attributed to the matched uncertainty assumption, we have the following lemma.   The proof is straightforward by following \cite{lopez2020adaptive-ccm} that considers matched parametric uncertainties and is thus omitted. 
\begin{lemma}\label{lemma:ccm_condition_uncertain_vs_nominal}
If a contraction metric $M(x)$ satisfies the strong 
CCM condition \cref{eq:ccm-condition-strong} for the nominal system, then the same metric satisfies the strong CCM condition
for the learned dynamics \cref{eq:dynamics-learned}, as well as for the uncertain system \cref{eq:dynamics-w-dnn} with the learned dynamics $\hat d(x)$.
\end{lemma}

 Define 
 \begin{equation}\label{eq:mcD-defn}
     \mcD \trieq \{y\in \mbR^m: \norm{ y}\leq b_d\}.
 \end{equation}
 Equation \cref{eq:mcD-defn} and Assumption~\ref{assump:lipschitz-bound-d-B} imply that $d(x)\in \mcD$ for any $x\in\mcX$.  As discussed in \cref{sec:sub-ccm-review},  when provided with a CCM and a feasible desired trajectory $\xstar(t)$ and $\ustar(t)$, 
a controller can be designed to exponentially stabilize the actual state trajectory $x(t)$ to the desired trajectory $\xstar(t)$. In practice, we only have access to learned dynamics to plan a trajectory. We now present a lemma providing the condition under which $\xstar(t)$ planned using the learned dynamics \cref{eq:dynamics-learned} also represents a feasible state trajectory for the true system. 
\begin{lemma}\label{lemma:feasible-traj-true-sys}
Considering a trajectory $(\xstar(\cdot), \ustar(\cdot))$ 
satisfying the learned dynamics \cref{eq:dynamics-learned} with $\xstar(t)\in \mcX$ for all $t\geq 0$, if 
\begin{equation}\label{eq:u-traj-condition-for-feasibility}
    \ustar(t)+\hat d(\xstar(t)) \in \mcU\ominus \mcD \quad  \forall t\geq 0,
\end{equation}
and Assumption~\ref{assump:lipschitz-bound-d-B} hold, then, 
 $\xstar(\cdot)$ is a feasible state trajectory for the true system \cref{eq:dynamics}.   
\end{lemma}
\begin{proof}
Note that $\bar u^\star(t) \trieq \ustar(t)-\tilde d(\xstar(t)) = \ustar(t)+\hat d(\xstar(t)) - (\hat d(\xstar(t))+\tilde d(\xstar(t))) = \ustar(t)+\hat d(\xstar(t)) - d(\xstar(t))$. Since  $\ustar(t)+\hat d(\xstar(t)) \in \mcU\ominus \mcD$ and $-d(\xstar(t))\in \mcD$, which is due to $\xstar(t)\in \mcX$ and Assumption~\ref{assump:lipschitz-bound-d-B}, we have $\bar u^\star(t)\in \mcU$. Comparison of \cref{eq:dynamics-learned} and \cref{eq:dynamics-w-dnn} implies that $\xstar(t)$ and $\bar u^\star(t)$ satisfies the true dynamics \cref{eq:dynamics-w-dnn}.  Thus, $\xstar(\cdot)$
is a feasible state trajectory for \cref{eq:dynamics-w-dnn}. \hfill $\blacksquare$ 
\end{proof}


\subsection{Filtered Disturbance Estimate}\label{sec:filter-dist}

Using a small $T$ in the estimation law \cref{eq:estimation_law} may introduce high-frequency components into the control signal, which could potentially compromise the robustness of the closed-loop system, e.g., against input delay. \blue{This has been demonstrated in the adaptive control literature, \cite[Sections 1.3 and 2.1.4]{naira2010l1book}, which shows that a high adaptation rate (corresponding to a  small $T$ here) will lead to a reduced time-delay margin}. 
To prevent the high-frequency signal in the estimation loop induced by small $T$ from entering the control loop, we can use a low-pass filter to smooth the estimated uncertainty before using it to compute control signals, \blue{as inspired by the \lonew adaptive control theory \cite{naira2010l1book}}. More specifically, we define the filtered disturbance estimate $\check d(t)$ as 
\begin{equation}\label{eq:check-d-filter-defn}
   {\check d(t)} = {\hat d(x(t))} +  \mathfrak{L}^{-1}\left[\mcC(s)\laplace{\check d^0(t) - \hat d(x(t))}\right], 
\end{equation}
where $\laplace{\cdot}$ and $\mathfrak{L}^{-1}[\cdot]$  denote the Laplace transform and inverse Laplace transform operators, respectively, and $\mcC(s)$ is a $m\times m$ strictly-proper transfer function matrix denoting a stable low pass filter. 
Notice that $\check d^0 (t)- \hat d(x(t))$ is an estimate of the learned model error $\tilde d(x(t)) = d (t)- \hat d(x(t))$. 
Filtering  $\hat d(x(t))$ is not necessary because it will not induce high-frequency signals into the control loop, unlike $\check d^0(x(t))$. 

For simplicity, we select    $\mcC(s)$ to be  
\begin{equation}\label{eq:filter-defn}
\begin{split}
\mcC(s) & = \textup{diag}({k_f^1}/{(s+ k_f^1)}, \dots, {k_f^m}/{(s+ k_f^m)}), 
\end{split}
\end{equation}
where  $k^j_f$ ($j=1,\dots,m$) is the filter bandwidth for the $j$th uncertainty channel. $\mcC(s)$ in \cref{eq:filter-defn} can be described by a state-space model $(A_f,B_f,I)$, where $A_f = \textup{diag}(-k_f^1,\dots,-k_f^m)$, and $B_f$ is a $m\!\times\! m$ matrix with all elements equal to $0$ except the $(j,j)$ element equal to $k_f^j$.  Define 
\begin{equation}
    {\underline{k}_f} = \min \{ k_f^1, \ldots k_f^m\}, \quad {{\bar k}_f} = \max \{ k_f^1, \ldots ,k_f^m\}.
\end{equation}
Leveraging the bound in \cref{eq:estimation_error_bound} and solution of state-space systems, we can straightforwardly derive an EEB on $\check d(t) - d(x(t))$, formally stated in the following lemma.

{\begin{lemma} Consider the filtered disturbance estimate given in \cref{eq:check-d-filter-defn}. If Assumption~\ref{assump:uniform-err-bnd} holds and condition \cref{eq:estimation_error_bound} holds for all $t$ in $[0,\xi]$, then,
\begin{align}
    \hspace{-6mm}  \norm{\check d(t)\!-\!d(x(t))} & \!\leq\!
    \bardeltade(t) \!\trieq\!\psi_1(t) \!+\! \sqrt{m}\bar\delta_{\tilde d}, \hspace{-3mm}\label{eq:estimation_error_bound_filtered}
\end{align} 
for all $t$ in $[0,\xi]$, where 
\begin{align}
      \psi_1(t) \!\trieq \! \left\{
    \begin{array}{ll}
  \!\!\! \psi_2 \trieq {\frac{{{b_d}{{\bar k}_f}}}{{{\underline{k}_f}}}( {1 - {e^{ - {\underline{k}_f}T}}})},  \ &\forall 0\leq t<T, \\
  \!\!\!  \psi_2 + \frac{{\bardeltade^0(T){{\bar k}_f}}}{{{\underline{k}_f}}}, \  &\forall t \!\geq\! T.\hspace{-5mm}
    \end{array}
    \right. \hspace{-3mm}\label{eq:psi1-defn}
\end{align}
\end{lemma}
\begin{proof}
Equation \cref{eq:check-d-filter-defn} implies 
\begin{align}
    \check d(t) - d(x(t)) = \Delta_1(t) + \Delta_2(t), \label{eq:df-d-Delta12-relation}
\end{align}
where \begin{align}
    \laplace{\Delta_1(t)} & =\mcC(s)\laplace{\check d^0(t)- d(x(t))}, \label{eq:Delta1-defn} \\
    \laplace{\Delta_2(t)} & = (I-\mcC(s))\laplace{\hat d(x(t))-d(x(t))}.\label{eq:Delta2-defn} 
\end{align}
Notice that \cref{eq:Delta1-defn} can be represented with a state-space model:
\begin{subequations}\label{eq:filter-dynamics}\vspace{-4mm}
\begin{align}
    \dot x_f(t) & = A_f x_f(t) + B_f (\check d^0(t)-d(x(t)) \label{eq:filter-dynamics-a}\\
     \Delta_1(t) & = x_f(t),\ x_f(0) = 0,\label{eq:filter-dynamics-b}
\end{align}
where $x_f(t)\in\mbR^m$ is the state vector of the filter.
\end{subequations}
From \cref{eq:filter-dynamics,eq:estimation_error_bound}, we have for any $t$ in $[0,T]$,
\begin{equation*}
\begin{split}
 \norm{ \Delta_1(t)}  & 
 \le  \!{b_d} \!\int_0^t  \! {\left\| {{e^{{A_f}(t - \tau )}}{B_f}} \right\|d\tau }  \le {b_d} \!\int_0^t  \! {{{\bar k}_f}{e^{ - {\underline{k}_f}(t - \tau )}}d\tau }  \\  &  =  \!{b_d}\frac{{{{\bar k}_f}}}{{{\underline{k}_f}}}\left( {1 - {e^{ - {\underline{k}_f}t}}} \right) 
 \le {b_d}\frac{{{{\bar k}_f}}}{{{\underline{k}_f}}}\left( {1 - {e^{ - {\underline{k}_f}T}}} \right),
\end{split}
\end{equation*}
where the second inequality is due to the fact that $e^{{A_f}(t - \tau )}{B_f}=\textup{diag}(k_f^1e^{-k_f^1},\dots,k_f^me^{-k_f^m})$. 
Additionally, for any $t$ in $[T,\xi]$, since $  {{\Delta_1}(t)}  \!=\!   {e^{{A_f}(t - T)}}x(T)\! + \!\int_T^t {{e^{{A_f}(t - \tau )}}{B_f}\left[{\check d^0(\tau ) \!-\! d(x(\tau ))} \right]d\tau }  $, we have
\begin{equation*}
\begin{split}
  \left\| {{\Delta_1}(t)} \right\| & \!\leq \!\left\| {x(T)} \right\| \!+\! \bardeltade^0(T)\!\int_T^t \!{\left\| {{e^{{A_f}(t - \tau )}}{B_f}} \right\|d\tau }  \\ 
    &\!\leq \!\psi_2 \!+\! \bardeltade^0(T)\underbrace {\int_T^t {{{\bar k}_f}{e^{ - {\underline{k}_f}(t - \tau )}}d\tau } }_{ = \frac{{{{\bar k}_f}}}{{{\underline{k}_f}}}\left( {1 - {e^{ - {\underline{k}_f}(t - T)}}} \right)} \leq \psi_2 \!+\! \frac{\bardeltade^0(T){{\bar k}_f}}{\underline{k}_f}. 
\end{split}
\vspace{-12mm}
\end{equation*}
As a result, we have 
\begin{equation}\label{eq:Delta1-bnd}
\norm{\Delta_1(t)}\leq \psi_1(t),\ \forall t\in[0,\xi].
\end{equation}
On the other hand, since $\infnorm{\hat d(x(t))\!-\!d(x(t))} \!\leq\! \norm{\hat d(x(t))\!-\!d(x(t))} \!=\! \norm{\tilde d (x(t))}\leq \bardeltatild$ for all $t$ in $[0,\xi]$ by assumption, by applying \cite[Lemma A.7.1]{naira2010l1book} to \cref{eq:Delta2-defn}, we have 
$\infnorm{\Delta_2(t)}\!\leq\! \lonenorm{I\!-\!\mcC(s)}\bardeltatild$ for all $t$ in $[0,\xi]$, where $\lonenorm{\cdot}$ denotes the \lonew norm (also known as the induced $\linf$ gain) of a system. Further considering  $\norm{\Delta_2(t)}\!\leq\! \sqrt{m}\infnorm{\Delta_2(t)}$, $\lonenorm{I-\mcC(s)} = 1$ (due to the specific $\mcC(s)$ selected in \cref{eq:filter-defn}),
\cref{eq:df-d-Delta12-relation} and \cref{eq:Delta1-bnd}, we obtain \cref{eq:estimation_error_bound_filtered}. \hfill $\blacksquare$ 
\end{proof}

The \lonew norm of a linear time-invariant system can be easily computed using the impulse response \cite[Section~A.7.1]{naira2010l1book}.
\begin{remark}
From the definitions of $\bardeltade^0(t)$ in \cref{eq:estimation_error_bound} and of $\psi_1(t)$ in \cref{eq:psi1-defn}, one can see that $\lim_{T\rightarrow0} \psi_1(t) = 0$, for all $t\geq T$. On the other hand, 
$\bardeltatild$ is expected to decrease with the improved accuracy of the learned model $\hat d(x)$. 
\end{remark}
\begin{remark}
As explained in \cref{remark:est_err_bnd_for_practical_implementation}, the bound $\bardeltade^0(t)$ could be quite conservative, which leads to a conservative bound, i.e., $\psi_1(t)$, for $\Delta_1(t)$ (defined in \cref{eq:Delta1-defn}). Additionally, the bound on $\Delta_2(t)$ can also be quite conservative due to the use of the \lonew norm of a system and the $\linf$ norms of inputs and outputs (e.g., characterized by \cite[Lemma~A.7.1]{naira2010l1book}).
As a result, the bound $\bardeltade(t)$ is most likely rather conservative. For practical implementation, one could leverage some empirical study, e.g., doing simulations
under a few user-selected functions of $d(x)$ and $\hat d(x)$ and determining
a tighter bound based on the simulation results. 
\end{remark}}
\begin{remark}
    \blue{For many control techniques, including a low-pass filter will induce phase delay and may damage system performance and/or robustness. However, this is not true for adaptive or uncertainty estimation-based control. As demonstrated in the adaptive control literature\cite{naira2010l1book}, including a low-pass filter with properly designed bandwidth can actually improve system robustness, e.g., against input delay. Simulation tests in \cref{sec:role_filter} indicate that this is also true for the proposed control technique. Additionally, the low-pass filter is only applied to the estimate of the learned model error (i.e., $\tilde d(x(t))$), and influences the estimation error bound. Such an influence will decrease when the learned model improves.}
\end{remark}
\subsection{Robust Riemann Energy Condition under Learned Dynamics}\label{sec:sub-robust-riem-energy-condition}
\cref{sec:sub-ccm-review} demonstrates that with a nominal system and a CCM established for it, one can devise a controller to limit the rate of decrease of the Riemann energy,  as described by \cref{eq:Edot-E-ineq}. Now, given uncertain dynamics with the learned model in  \cref{eq:dynamics-w-dnn}, and the planned trajectory $(\xstar(\cdot), \ustar(\cdot))$ using the learned model, under the condition \cref{eq:u-traj-condition-for-feasibility},  the condition \cref{eq:Edot-E-ineq} now becomes 
\begin{equation}
    \label{eq:Edot-E-ineq-detailed-under-uncertainty}
\hspace{-2mm}\gamma_s^\top\!(1,t)M(x(t))\dot x(t) -  \gamma_s^\top\!(0,t)M(\xstar(t)) F_l(\xstar(t)\!,\ustar(t))
\!\leq\! \!- \!\lambda E(\xstar(t)\!,x(t)), \hspace{-1mm}
\end{equation}
where     $\dot{x}(t) = f(x(t))+B(x(t))(u(t) + d(x(t)))$ denotes the true dynamics evaluated at $x(t)$, and 
$F_l(x,u)$ denotes the learned dynamics defined in \cref{eq:Fl-defn}. Obviously, \cref{eq:Edot-E-ineq-detailed-under-uncertainty} is {\it not implementable} since it depends on the uncertainty $d(x(t))$. 
However, 
by estimating the value of $d(x(t))$ at each time $t$ with a computable estimation error bound, we could establish a robust condition for \cref{eq:Edot-E-ineq-detailed-under-uncertainty} \cite{zhao2022robust-ccm-de}. In particular, consider the filtered disturbance estimate introduced in \cref{sec:filter-dist} that estimates $d(x(t))$ as $\check d(t)$ at each time $t$ with an error bound $\bardeltade(t)$ defined in \cref{eq:estimation_error_bound_filtered}.  
Then, a sufficient condition for \cref{eq:Edot-E-ineq-detailed-under-uncertainty} can be obtained as 
{
\begin{align}
    \gamma_s^\top\!(1,t)M(x)\dot{\check x}-  \gamma_s^\top\!(0,t)M(\xstar)F_l(\xstar,\ustar) +  \norm{\gamma_s^\top\!(1,t)M(x)B(x)}\bardeltade(t) \!\leq \!- \lambda E(\xstar,x), 
\label{eq:Edot-E-ineq-detailed-robust}
\end{align}}where 
$\dot{\check x} (t)\trieq  f(x)+B(x)(u(t) +\check d(t)).
$ Moreover, since $M(x)$ satisfies the CCM condition \cref{eq:ccm-condition-strong}, control input $u(t)$ that satisfies \eqref{eq:Edot-E-ineq-detailed-robust} will exist  for any $t\geq 0$, irrespective of the size of $\delta$, \blue{in the absence of} control limits, i.e., $\mcU =\mbR^m$. We call condition \cref{eq:Edot-E-ineq-detailed-robust} the {\it robust Riemann energy} (RRE) condition.
\begin{remark}
From \cref{sec:sub-robust-riem-energy-condition,sec:sub-disturbance-estimation}, we see that the disturbance estimate provided by \cref{eq:state-predictor,eq:estimation_law} and incorporated in the RRE condition \cref{eq:Edot-E-ineq-detailed-robust} is the discrepancy between the true dynamics and the nominal dynamics without the learned model (i.e., $d(x)$). 
 Alternatively, we can easily adapt the estimation law \cref{eq:state-predictor,eq:estimation_law} to estimate $\tild(x)$ (i.e., the discrepancy between learned dynamics \cref{eq:dynamics-learned} and true dynamics), and adjust the RRE condition accordingly. However, 
as characterized in \cref{eq:estimation_error_bound}, the EEB depends on the local Lipschitz bound of the uncertainty to be estimated, which indicates that a Lipschitz bound for $\tild(x)$ is needed to establish the EEB for it. 
\end{remark}
\subsection{Guaranteed Trajectory Tracking Under Learned Dynamics}\label{sec:robust-contraction-DE}
Similar to \cref{sec:sub-ccm-review}, we can compute a control input at each time $t$ to satisfy the RRE condition \cref{eq:Edot-E-ineq-detailed-robust}. In practice, one may want to compute $u(t)$ with a minimal $ \norm{u(t)-\ustar(t)}$, which can be achieved by setting $u(t) = \ustar(t) + k^\ast(t,\xstar,x)$  with $k^\ast(t,\xstar,x)$ obtained via solving a QP problem:
\begin{align} 
   & k^\ast(t,\xstar,x)=\argmin_{k\in\mbR^m} \norm{k}^2   \label{eq:qp-robust} \\
\textup{ s.t. } &   \phi_0(t,\xstar,x) + \phi_1^\top(\xstar,x)k\leq 0
   \label{eq:robust-reim-energy-condition-final}
\end{align}
at each $t\geq 0$, where 
\begin{align}
    \phi_0(t,\xstar,x) \trieq & ~\gamma_s^\top(1,t)M(x)( f(x)+B(x)(\ustar(t)+\check d(t)))+ \nonumber \\
    & ~\norm{\gamma_s^\top\!(1,t)M(x)B(x)}\bardeltade(t)  -  \gamma_s^\top(0,t)M(\xstar) F_l(\xstar\!,\ustar)+\lambda E(\xstar,x)\\
    \phi_1(\xstar,x)  \trieq &~B^\top(x)M(x)\gamma_s(1,t),\label{eq:phi1-defn}
\end{align}
and
\cref{eq:robust-reim-energy-condition-final} is an equivalent representation of \cref{eq:Edot-E-ineq-detailed-robust}, 
$\check d(t)$ is the filtered disturbance estimate via \cref{eq:check-d-filter-defn,eq:state-predictor,eq:estimation_law,eq:dcheck-hsigma-relation}, $\bardeltade(t)$ is defined in \cref{eq:estimation_error_bound_filtered}, and $F_l(\cdot,\cdot)$ is defined in \cref{eq:Fl-defn}. The problem \cref{eq:qp-robust} 
is commonly referred to as a min-norm problem and can be analytically solved by \cite{freeman2008robust-clf-book} 

\begin{equation*}
    k^\ast \! = \!\left\{
    \begin{split}
      & 0 \quad & & \textup{ if } \phi_0(t,\xstar,x)\!\leq \!0, \\
      & -\!{\textstyle \frac{\phi_0(t,\xstar,x)\phi_1(\xstar,x)}{\norm{\phi_1(\xstar,x)}^2} }& & \textup{ else }  
    \end{split}
    \right.
\end{equation*}

\begin{remark}
The proposed controller is inspired by the \lonew adaptive control theory \cite{naira2010l1book}. In fact, we adopt the estimation mechanism (the PWCE law in \cref{eq:state-predictor,eq:estimation_law}) used within an \lonew controller. However, instead of directly canceling the estimated disturbance as one would do with an \lonew controller, the proposed approach incorporates the estimated uncertainty and the EEBs into the robust Riemann energy condition \cref{eq:robust-reim-energy-condition-final} to compute 
the control signal, which ensures exponential convergence of actual trajectories to desired ones. 
\end{remark}


{The following lemma establishes a bound on the control input given by solving \cref{eq:qp-robust}. 
\begin{lemma}\label{lemma:u-bnd-uncertain}
Suppose Assumptions~\ref{assump:uniform-err-bnd}--\ref{assump:ccm-exists-for-nominal-sys} hold, and the uncertainty estimate error is bounded according to \cref{eq:estimation_error_bound_filtered}. 
Moreover, suppose  $\xstar,~x\in \mcX$ such that $\mathtt d(\xstar,x)\leq \bar {\mathtt d}$ for a scalar constant $\bar {\mathtt d}\geq 0$. Then, the control effort from solving \cref{eq:qp-robust} is bounded as \vspace{-5mm}
\begin{equation}\label{eq:k-bnd-actual}
    \norm{ k^\ast(t,\xstar,x)}\leq \overbrace{\large \frac{\bar {\mathtt d}}{2}\sup_{x\in\mcX}{\frac{\bar \lambda(L^{T}\hat GL)}{\underline \sigma_{> 0}(B^T L^{-1})}}+\bar \delta_{\tilde d} +2\bardeltade(t)}^{\trieq \bar \delta_u^L}, 
\end{equation}where 
$\hat{G} (x) \!\trieq\!\left \langle \!M\frac{\partial \hat f}{\partial x}\!\right\rangle \!+ \!\partial_{\hat f} M \!+\!2\lambda M $ with $\hat f$ defined in \cref{eq:Fl-defn}.  
\end{lemma}
\begin{proof}
Due to Assumption~\ref{assump:uniform-err-bnd}, $M(x)$ is a valid CCM for the learned dynamics \cref{eq:dynamics-learned}, according to \cref{lemma:ccm_condition_uncertain_vs_nominal}. By applying \cref{lemma:u-bnd-nominal} to the learned dynamics \cref{eq:dynamics-learned}, we can obtain that at any $t\geq 0$, for any feasible $\xstar(t)$ and $x(t)$ satisfying $\mathtt d(\xstar,x)\leq \bar {\mathtt d}$,
there  always exists a $k_0$ subject to $\norm{k_0}\leq \frac{\bar {\mathtt d}}{2}\sup_{x\in\mcX}{\frac{\bar \lambda(L^{T}\hat GL)}{\underline \sigma_{> 0}(B^T L^{-1})}}$  such that \begin{equation}\label{eq:learned-k0-ineq}
  \gamma_s^\top\!(1,t)M(x)\left(F_l(x,\ustar)+\! k_0 \right) \!-  \gamma_s^\top\!(0,t)  M(\xstar)F_l(\xstar,\ustar)
   \!   + \lambda E(\xstar,x) \leq\! 0.
\end{equation}
 Setting 
 \begin{equation}\label{eq:k-defn}
     k=k_0+\hat d(x) - \check d(t) - 
 {\phi_1 \bardeltade(t)}/{\norm{\phi_1}}
 \end{equation}
 with $\phi_1$ defined in \cref{eq:phi1-defn}, we have 
 \begin{align}
  \hspace{-6mm}  \textup{LHS of }\cref{eq:robust-reim-energy-condition-final} & = \textup{LHS of }\cref{eq:learned-k0-ineq} \! \underbrace{-\phi_1^T \phi_1 {\textstyle \frac{ \bardeltade(t)}{\norm{\phi_1}}} \!+\! \norm{\phi_1} \bardeltade(t)}_{=0} \leq 0 \hspace{-2mm} \nonumber 
  \vspace{-2mm}
 \end{align}
 which implies that $k$ defined in \cref{eq:k-defn} is a feasible solution for \cref{eq:qp-robust}. Therefore, the optimal solution  for \cref{eq:qp-robust} satisfies 
   $\norm{ k^\ast(t,\xstar,x)} \leq  \norm{ k}    = \norm{k_0\!+\!\hat d(x) \!-\! d(x) \!+ \!d(x)\!- \!\check d(t) \!- \!\phi_1 { \bardeltade(t)}/{\norm{\phi_1}}} 
   \leq  \norm{k_0} \!+ \!\|{\hat d(x) \!-\! d(x)}\| \!+\! \|{d(x)\!- \!\check d(t)}\| \!+\! \bardeltade(t) 
   \leq \frac{\bar {\mathtt d}}{2}\sup_{x\in\mcX}{\frac{\bar \lambda(L^{T}\hat GL)}{\underline \sigma_{> 0}(B^T L^{-1})}}\!+\!\bar \delta_{\tilde d} \!+\!2\bardeltade(t),$
which proves \cref{eq:k-bnd-actual}. \hfill $\blacksquare$ 
\end{proof}}


{\color{black} The main theoretical result of the paper is stated below.
\begin{theorem} \label{them:DE-CCM}
Consider an uncertain system \cref{eq:dynamics-w-dnn} with learned dynamics  $\hat d(x)$. Suppose Assumptions~\ref{assump:lipschitz-bound-d-B}--\ref{assump:uniform-err-bnd} hold and Assumption~\ref{assump:ccm-exists-for-nominal-sys} holds with positive constants $\alpha_1$,  $\alpha_2$ and $\lambda$. Furthermore, suppose the initial state vector $x(0)\in\mcX$ and  a continuous trajectory ($\xstar(\cdot),~\ustar(\cdot)$) planned using the learned dynamics \cref{eq:dynamics-learned} satisfy \cref{eq:u-traj-condition-for-feasibility,eq:ustar-bardeltauL-condition,eq:xstar0-x0-condition},
{
\begin{align}
& \ustar(t) \in \mcU\ominus \left\{y\in\mbR^m: \norm{y}\leq \bar\delta_u^L \right\}, \label{eq:ustar-bardeltauL-condition} \\
 \hspace{-3mm}  \Omega(t)\!\trieq\! & {\left \{y\!\in \!\mbR^n\!: \norm{y} \leq \norm{\xstar(t)}\! +\! \!\sqrt{\textstyle \frac{\alpha_2}{\alpha_1}}\norm{x(0)\!-\!\xstar(0)}e^{-\lambda t} \right \}} \subset \textup{int}(\mcX), \label{eq:xstar0-x0-condition}
\end{align}}for any $t\geq 0$, where $\textup{int}(\cdot)$ denotes the interior of a set, $ \bar \delta_u^L $ is defined in \cref{eq:k-bnd-actual} with 
\begin{equation}\label{eq:bard-defn}
\bar{\mathtt d} = \mathtt d(\xstar(0),x(0))e^{-\lambda t}+\varepsilon
\end{equation}
for an arbitrary $\varepsilon>0$. Then, 
the control law $u(t) = \ustar(t)+k^\ast(t,\xstar,x)$ with $k^\ast(t,\xstar,x)$ solving \cref{eq:qp-robust} ensures that $u(t)\in \mcU$ and $x(t)\in \mcX$ for all $t\geq 0$, 

and universally exponentially stabilizes the  system \cref{eq:dynamics-w-dnn} in the sense that \cref{eq:ues-math} holds. \end{theorem}
\begin{proof} 
Since Assumption~\ref{assump:uniform-err-bnd} holds, we have $\tild(x(0))\leq \bardeltatild$ since $x(0)\in\mcX$. The EEB in \cref{eq:estimation_error_bound_filtered}  holds for  $t=0$  due to the fact that $\norm{\check d(0) - d(x(0)) } = \norm{\hat d(0) - d(x(0))} = \norm{\tilde d(0)} \leq \bardeltatild \leq \bardeltade(0)$. 
As a result, the bound in \cref{eq:k-bnd-actual} holds for $t=0$ with $\bar{\mathtt d} = \mathtt d(\xstar(0),x(0))+\varepsilon$, which, together with \cref{eq:ustar-bardeltauL-condition}, implies $u(0) = \ustar(0) + k^\ast(0,\xstar(0),x(0))\in \mcU$.

 We next prove $x(t)\in \mcX$ and $u(t)\in\mcU$ for all $t\geq 0$ by contradiction.  Assume this is not true. Since $x(t)$ and $u(t)$ are continuous\footnote{$u(t)=\ustar(t)+k^\ast(t,\xstar,x)$ is continuous because $\ustar(t)$ is continuous, and $k^\ast(t,\xstar,x)$ is also continuous according to \cref{eq:qp-robust}.} 
 there must exist a time $\tau$ such that
\begin{align}
 x(t) &\in \mcX,\quad u(t)\in \mcU,  \  \forall t\in[0,\tau), \label{eq:xu-in-XU-0-tau^-} \\
    x(\tau) &\notin \mcX \textup{ or } u(\tau)\notin \mcU. \label{eq:x-tau-notin-X-assump}
\end{align}
Now, examine the evolution of the system within the interval  $[0,\tau)$. Due to \cref{eq:xu-in-XU-0-tau^-},  the error bound in \cref{eq:estimation_error_bound} holds within $[0,\tau)$. Also, notice that 
condition \cref{eq:u-traj-condition-for-feasibility} ensures that $\xstar(\cdot)$ represents a feasible trajectory for the uncertain system \cref{eq:dynamics-w-dnn} with input constraints according to \cref{lemma:feasible-traj-true-sys}. Therefore, the control law $u(t) = \ustar(t) + k^\ast(t,\xstar,x)$ given by \cref{eq:qp-robust} 
 guarantees the fulfillment of of the RRE condition \cref{eq:Edot-E-ineq-detailed-robust}, and consequently, of condition \cref{eq:Edot-E-ineq-detailed-under-uncertainty}, which further implies \cref{eq:Edot-E-ineq}, or equivalently,
\begin{equation}\label{E-E0-relation}
    E(\xstar(t),x(t))\leq E(\xstar(0),x(0))e^{-2\lambda t},
\end{equation}
for all $t$ in $[0,\tau)$. \cref{them:synthesis-nominal} and \cref{E-E0-relation} indicate that $\norm{x(t)} \leq \norm{\xstar(t)}+\sqrt{\frac{\alpha_2}{\alpha_1}} \norm{x(0)-\xstar(0)} e^{-\lambda t}$ $\forall t\in [0,\tau)$, which, together with \cref{eq:xstar0-x0-condition},  implies that $x(t)$ stays in the interior of  $\mcX$,  $\forall t\in[0,\tau)$. Considering that $x(t)$ is continuous, we have $x(\tau) \in \mcX$, 
contradicting the first condition in \cref{eq:x-tau-notin-X-assump}. As a result, we have 
\begin{equation}\label{eq:x-in-X-0tau}
    x(t)\in \mcX,\  \forall t\in[0,\tau].
\end{equation}
Now consider the second condition in \cref{eq:x-tau-notin-X-assump}. Condition \cref{E-E0-relation} implies  $ \mathtt d(\xstar(t),x(t))\leq \mathtt d (\xstar(0),x(0))e^{-\lambda t}<\bar {\mathtt d}$ for any $t$ in $[0,\tau)$, where $\mathtt d(\xstar,x)$ denotes the Riemann distance between $\xstar$ and $x$ and $\bar {\mathtt d}$ is defined in \cref{eq:bard-defn}. Considering the continuity of $\xstar(t)$ and $x(t)$, we have 
\begin{equation}\label{eq:d-bnd-0-tau}
    {\mathtt d}(\xstar(t),x(t))\leq \bar {\mathtt d},\ \forall t\in[0,\tau].
\end{equation}
Due to \cref{eq:x-in-X-0tau} and $u(t)\in \mcU$ for all $t\in[0,\tau)$ (from \cref{eq:xu-in-XU-0-tau^-}), it follows from \cref{lemma:estiamte-error-bound} that the error bound in \cref{eq:estimation_error_bound} holds in $[0,\tau]$, which, together with \cref{eq:x-in-X-0tau}, implies the filter-dependent error bound in \cref{eq:estimation_error_bound_filtered} holds  in $[0,\tau]$. This fact, along with \cref{eq:d-bnd-0-tau}, indicates that \cref{eq:k-bnd-actual} holds, i.e., $\norm{k^\ast(t,\xstar,x)}\leq  \bar \delta_u^L $, for all $t$ in $[0,\tau]$. Further considering \cref{eq:ustar-bardeltauL-condition}, we have $u(t)=\ustar(t)+k^\ast(t,\xstar,x)\in \mcU$ for all $t$ in $[0,\tau]$, which, together with \cref{eq:x-in-X-0tau}, contradicts \cref{eq:x-tau-notin-X-assump}. 
Therefore, we conclude that $x(t)\in \mcX$ and $u(t)\in \mcU$ for all $t\geq 0$. From the development of the proof, it is clear that the the UES of the closed-loop system in $\mcX$ with the control law given by the solution of \cref{eq:qp-robust} is achieved. 
\hfill $\blacksquare$
\end{proof}}

\begin{remark}
Theorem~\ref{them:DE-CCM} asserts that under specific conditions, the proposed controller ensures {\em exponential convergence} of the actual state trajectory to a desired trajectory that may be planned using a potentially poorly learned model. 
{On the other hand, the improved accuracy of the learned model reduces the error bound $\bardeltatild$ and thus the EEB $\bardeltade(t)$, and improve the robustness of the controller, as will be demonstrated by simulation results in \cref{sec:sim}.}
\end{remark}

\begin{remark}\label{remark:discrete-implementation-issue}
The exponential convergence assurance described in Theorem \ref{them:DE-CCM} relies on a continuous-time implementation of the controller. However, in practical applications, controllers are typically implemented on digital processors using fixed sampling times. Consequently, the property of exponential convergence may be slightly compromised, as noted in Section \ref{sec:sim}.
\end{remark}

\subsection{Discussions}
With consideration of \cref{eq:k-bnd-actual}, condition \cref{eq:ustar-bardeltauL-condition} in \cref{them:DE-CCM} requires that when planning nominal input trajectories, enough control authority should be left for the control law defined by \cref{eq:qp-robust}. Additionally, the control effort bound $\bar\delta_u^L$ depends on the error bound of the learned model,  $\bardeltatild$; a poorly learned model with  large $\bardeltatild$ will lead to large $\bar\delta_u^L$, making it challenging to satisfy \cref{eq:ustar-bardeltauL-condition}.

    While \cref{them:DE-CCM} only mentions the trajectory tracking performance, we will empirically show the benefits of learning in facilitating better trajectory planning and improving the robustness of the controller in Section~\ref{sec:sim}.

\section{Simulation Results}\label{sec:sim}
We validate the proposed learning control approach on a 2D quadrotor introduced in \cite{singh2019robust}.  \blue{We selected this example because the 2D quadrotor, while simpler than a 3D quadrotor, presents significant control challenges due to its nonlinear and unstable dynamics. Additionally, it offers a suitable scenario for demonstrating the applicability of the proposed control architecture, specifically, maintaining safety and reducing energy consumption in the presence of disturbances.}. 
The dynamics of the vehicle are given by
    {\setlength{\arraycolsep}{1pt}
{\begin{equation*}
  \dot x =  \begin{bmatrix}
    \dot p_x \\
    \dot p_z \\
    \dot \phi \\
    \dot v_x \\
    \dot v_z \\
    \ddot \phi 
    \end{bmatrix} =
    \begin{bmatrix}
    v_x \cos(\phi)-v_z \sin(\phi)\\
    v_x\sin(\phi) + v_z \cos(\phi) \\
    \dot \phi \\
    v_z\dot \phi - g\sin(\phi) \\
    -v_x\dot \phi - g\cos(\phi) \\
    0  
    \end{bmatrix} +
    \begin{bmatrix}
    0 & 0 \\
    0 & 0 \\
    0 & 0 \\
    0 & 0 \\
    \frac{1}{m} & \frac{1}{m} \\
    \frac{l}{J} & -\frac{l}{J}
    \end{bmatrix}(u +d(x)),
\end{equation*}}}where $p_x$ and $p_z$ represent the positions in $x$ and $z$ directions, respectively, $v_x$ and $v_z$ denote the lateral velocity and velocity along the thrust axis in the body frame. $\phi$ is the angle between the $x$ direction of the body frame and the $x$ direction of the inertia frame. The input vector $u=[u_1,u_2]$ contains the thrust force produced by the two propellers. $m$ and $J$ represent the mass and moment of inertia about the out-of-plane axis, respectively. $l$ denotes the distance between each propeller and the vehicle center, and $d(x)$ signifies the unknown disturbances exerted on the propellers. Specific parameter values are assigned as follows: $m=0.486$ kg, $J=0.00383~\textup{Kg m}^2$, and $l=0.25$ m. 
We choose $d(x)$ to be $d(x) = \rho(x,z)\cdot 0.5(v_x^2+v_z^2)[1,1]^\top$, where $\rho(x,z) = 1/({(x-5)^2+(y-5)^2+1})$ represents the disturbance intensity whose values in a specific location are denoted by the color at this location in \cref{fig:results-diff-learning}. 
We consider three navigation tasks with different start and target points, 
while avoiding the three circular obstacles as illustrated in \cref{fig:results-diff-learning}.  The planned trajectories were computed using OptimTraj \cite{kelly2017intro-traj-opt} to minimize a cost function $J=\int_0^{T_a}\norm{u(t)}^2dt+5T_a$, where $T_a$ denotes the arrival time. The actual start points for Tasks~1$\sim$3 were  
intentionally set to be different from the desired ones used for trajectory planning. 

\subsection{Control Design}
For computing a CCM, we parameterized the CCM $W$ by $\phi$ and $v_x$,  and set the convergence rate $\lambda$ to 0.8.  Additionally, 
we enforced the constraints: $x\in \mcX \trieq [0,15] \times [0,15] \times [-\frac{\pi}{3},\frac{\pi}{3}] \times  [-2,2] \times [-1,-1] \times[-\frac{\pi}{3},\frac{\pi}{3}]$,  $u\in \mcU\trieq [0, \frac{3}{2}mg] \times [0, \frac{3}{2}mg]$.  
More details about synthesizing the CCM and computing the geodesic can be found in \cite{zhao2022tube-rccm-ral}. All the subsequent computations and simulations except DNN training (which was done in Python using PyTorch) were done in Matlab R2021b.
For estimating the disturbance using \cref{eq:state-predictor,eq:estimation_law,eq:dcheck-hsigma-relation}, we set $a=10$. It is straightforward to confirm  that $L_d = 4$, $b_d = 3.54$, and $L_B=0$ (due to the constant nature of $B$) satisfy \cref{eq:lipschitz-cond-d-B}. By discretizing the space $\mcX$ into a grid, one can determine the constant $\phi$ in \cref{lemma:estiamte-error-bound} to be $\phi= 783.96$. Based on \cref{eq:estimation_error_bound}, to achieve an error bound $\bar \delta_\textrm{de}^0=0.1$, the estimation sampling time needs to satisfy $T_s\leq 2.04\times 10^{-7} $ s.  However, as mentioned in Remark~\ref{remark:est_err_bnd_for_practical_implementation}, the error bound calculated from \cref{eq:estimation_error_bound} might be overly conservative. Through simulations, we determined that a sampling time of $0.002$~s was sufficient to achieve the desired error bound and thus set $T_s=0.002$ s. 

\begin{figure*}[tb]
     \centering
     \begin{subfigure}[b]{0.32\textwidth}
         \centering
         \includegraphics[width=\textwidth]{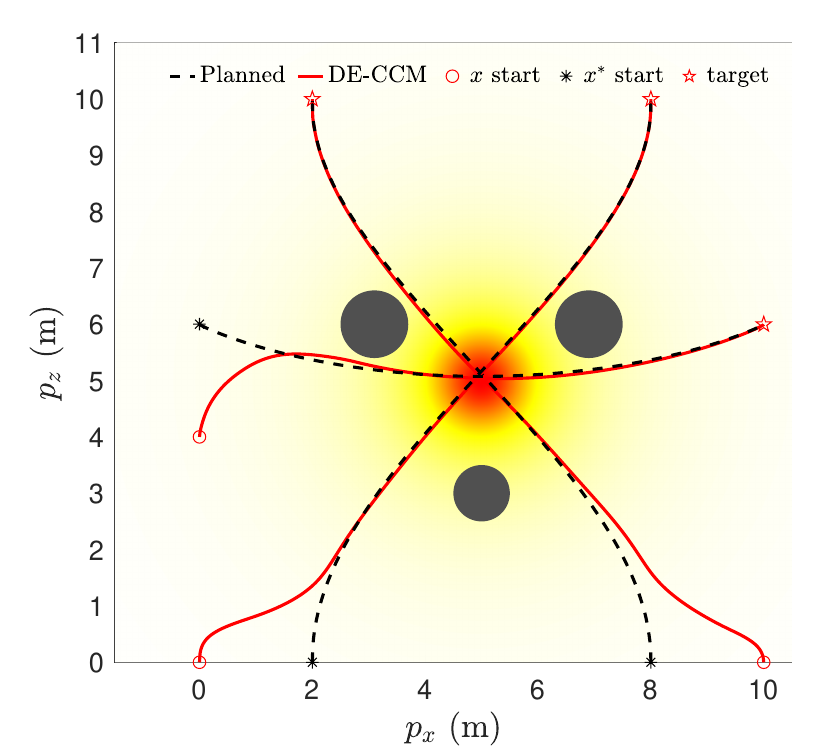}
     \end{subfigure}
     \begin{subfigure}[b]{0.32\textwidth}
         \centering
         \includegraphics[width=\textwidth]{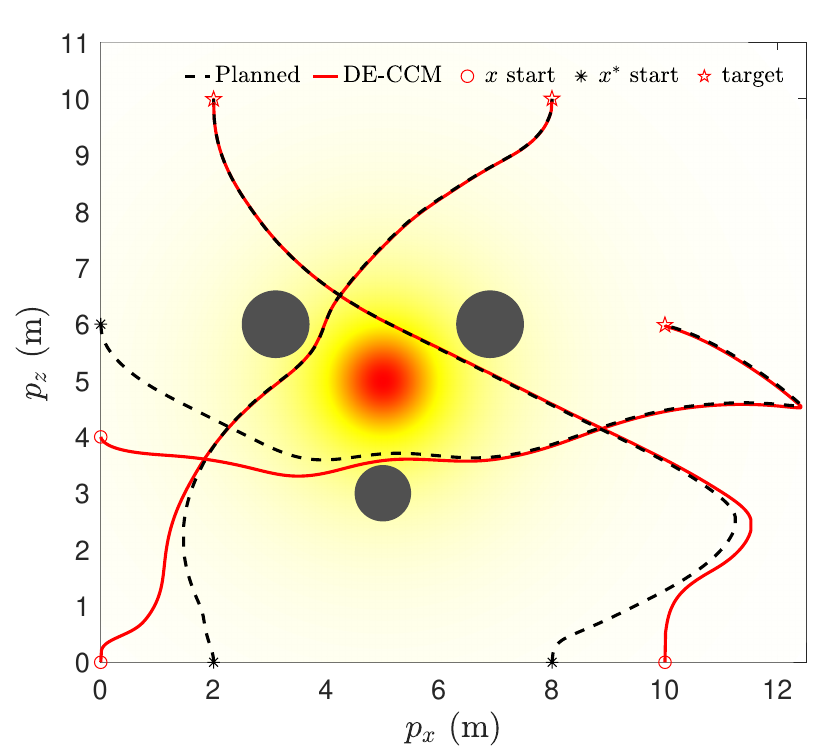}
     \end{subfigure}
     \begin{subfigure}[b]{0.32\textwidth}
         \centering
         \includegraphics[width=\textwidth]{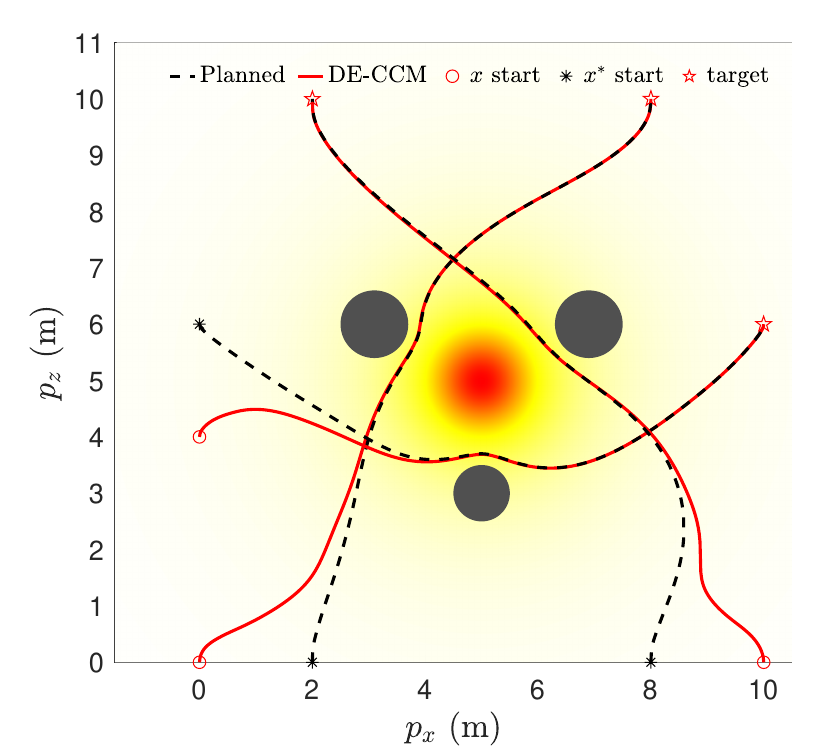}
     \end{subfigure} \vspace{-3mm} \\
      \begin{subfigure}[b]{0.32\textwidth}
         \centering
         \includegraphics[width=\textwidth]{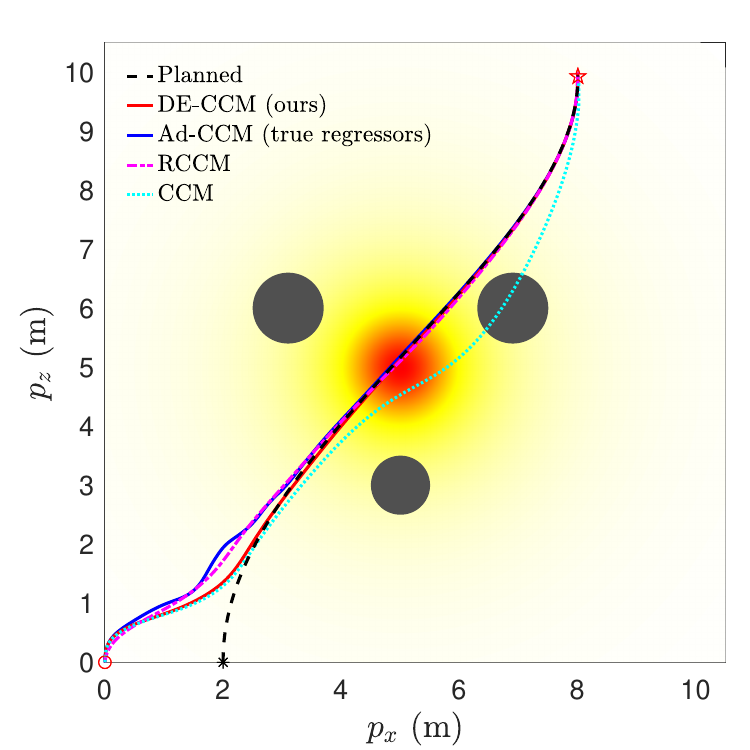}
         \caption{No Learning}
     \end{subfigure}
     \begin{subfigure}[b]{0.32\textwidth}
         \centering
         \includegraphics[width=\textwidth]{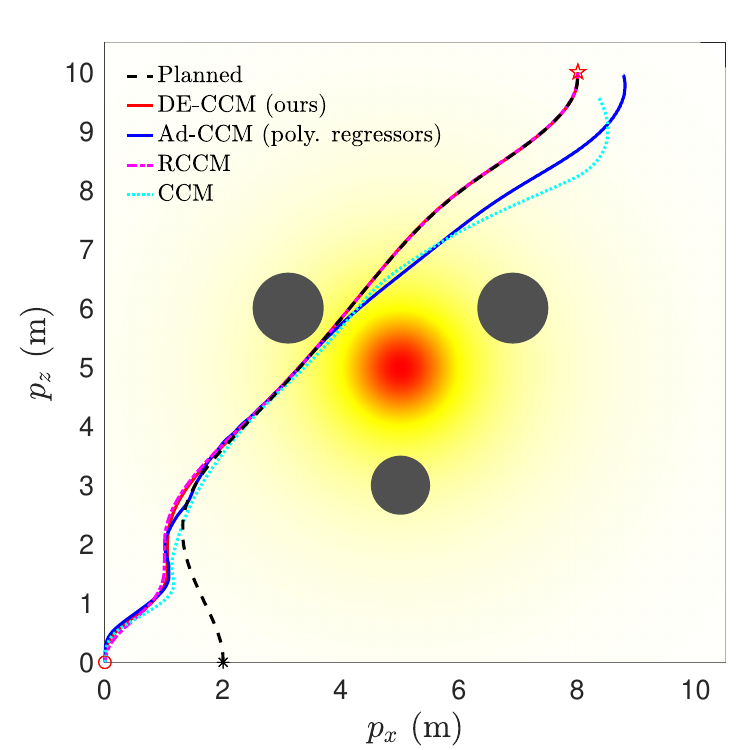}
         \caption{Moderate Learning}
     \end{subfigure}
     \begin{subfigure}[b]{0.32\textwidth}
         \centering
         \includegraphics[width=\textwidth]{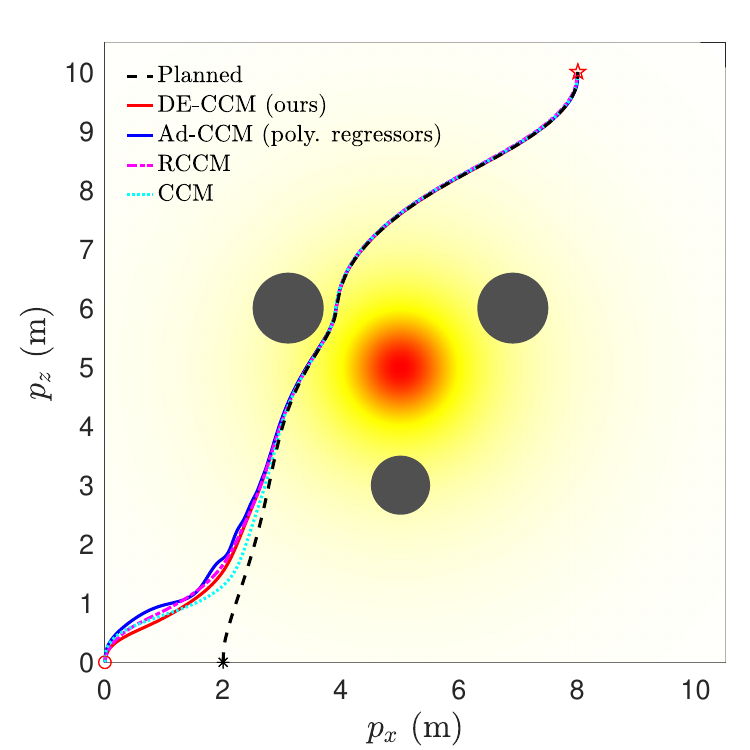}
         \caption{Good Learning}
     \end{subfigure}     
        \caption{Top: Planned and executed trajectories under our proposed DE-CCM for Tasks 1$\sim$3  under different learning scenarios. Bottom: Tracking performance of DE-CCM, Ad-CCM \cite{lopez2020adaptive-ccm}, RCCM (which can be seen as a special case of DE-CCM with uncertainty estimate set to $0$), and CCM for Task 1 under different learning scenarios
        Start points of planned and actual trajectories are intentionally set to be different to reveal the transient response under different scenarios.
        }
        \label{fig:results-diff-learning}
\end{figure*}
  
  \subsection{Performance and Robustness Across Learning Transients}
  Figure~\ref{fig:results-diff-learning} (top) illustrates the planned and actual trajectories under the proposed controller utilizing the RRE condition and disturbance estimation (referred to as \textbf{DE-CCM}), in the presence of {\it no, moderate} and {\it good} learned model for uncertain dynamics. For these results, we did {\it not}  low-pass filter the estimated uncertainty, which is equivalent to setting $\mcC(s) = I$ in \cref{eq:check-d-filter-defn}. 
  Spectral-normalized DNNs \cite{miyato2018spectral-normalization-abbrev} (\cref{remark:uniform-err-bnd}) with four inputs, two outputs and four hidden layers were used for model learning. 
For training data collection, we planned and executed nine trajectories with different start and end points, as shown in \cref{fig:safe_explore}. 
The data collected during the execution of these trajectories were used to train the {\it moderate} model. However, these nine trajectories were still not enough to fully explore the state space. 
 Sufficient exploration of the state space is necessary to learn a good uncertainty model. Thanks to the performance guarantee, the DE-CCM controller facilitates safe exploration, as demonstrated in \cref{fig:safe_explore}. 
For illustration purposes, we directly used the true uncertainty model to generate the data and used the generated data for training, which yielded the {\it good} model. 
\begin{figure}[ht]
    \centering
    \includegraphics[width=0.34\textwidth]{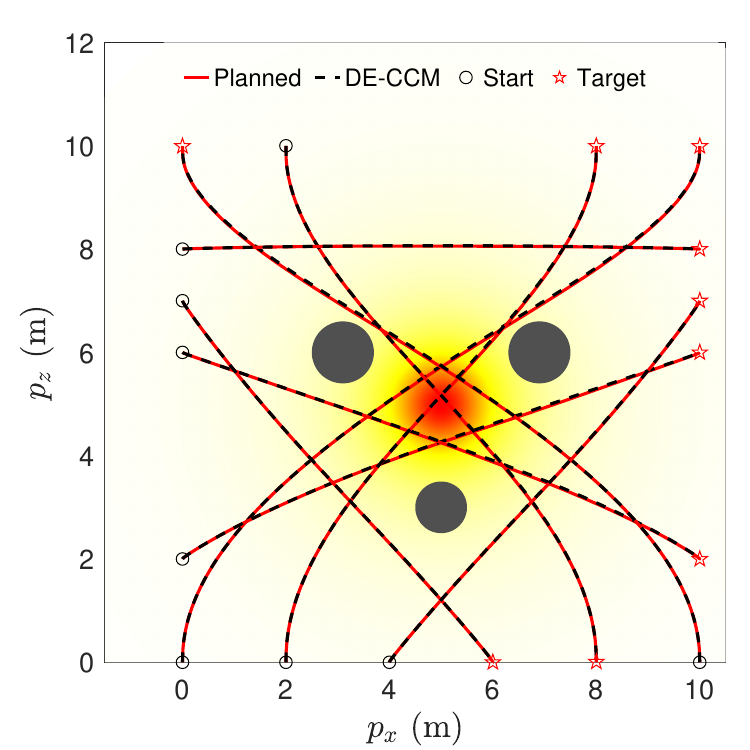}
    \caption{Safe exploration of the state space for learning uncertain dynamics under DE-CCM}
    \label{fig:safe_explore}
\end{figure}

As depicted in \cref{fig:results-diff-learning}, the actual trajectories generated by DE-CCM  exhibited expected convergence towards the desired trajectories  during the learning phase across all three tasks. The minor deviations observed between the actual and desired trajectories under DE-CCM  can be attributed to the finite step size used in the ODE solver employed for simulations (refer to Remark \ref{remark:discrete-implementation-issue}).
The planned trajectory for Task~2 in the moderate learning case seemed weird near the end point. This is because the learned model was not accurate in the area due to a lack of sufficient exploration. However, with the DE-CCM controller, the quadrotor was still able to track the trajectory. 
Figure~\ref{fig:disturbance} depicts the trajectories of true, learned and estimated disturbances in the presence of no and good learning for Task~1, while the trajectories for Tasks~2 and 3  are similar and thus omitted. 
One can see that the estimated disturbances were always fairly close to the true disturbances. Also, the area with high disturbance intensity was avoided under good learning, which explains the smaller disturbance encountered. 
\begin{figure}[ht]
    \centering
    \includegraphics[width=0.35\textwidth]{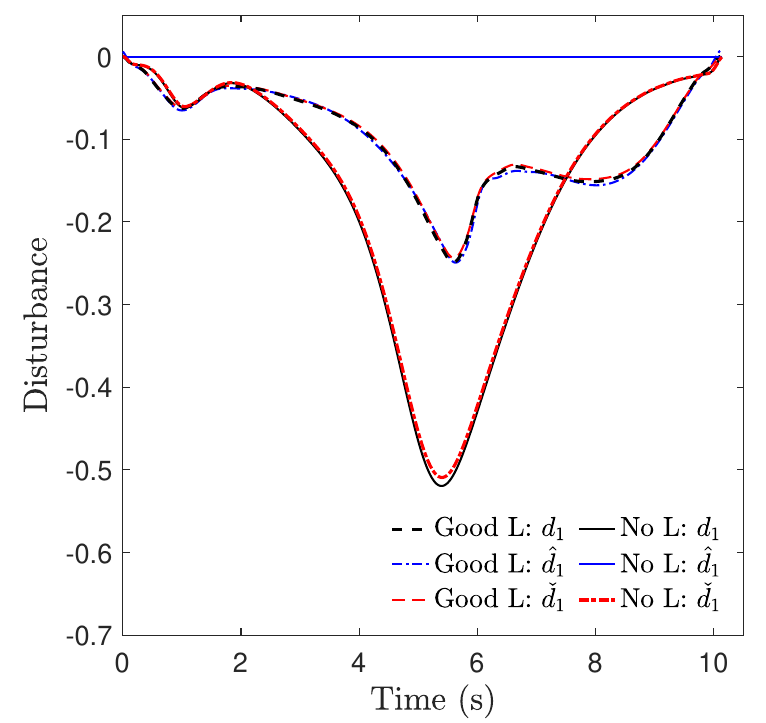}
    \caption{Trajectories of true, learned, and estimated disturbances in the presence of no and good learning for Task~1. The notations $d_1$, $\hat d_1$ and $\check d_1$ denote the first element of $d$, $\hat d$ and $\check d$, respectively.}
    \label{fig:disturbance}
\end{figure}

\begin{figure}[h]
     \centering
     \begin{subfigure}[b]{0.475\textwidth}
         \centering
         \includegraphics[width=\textwidth]{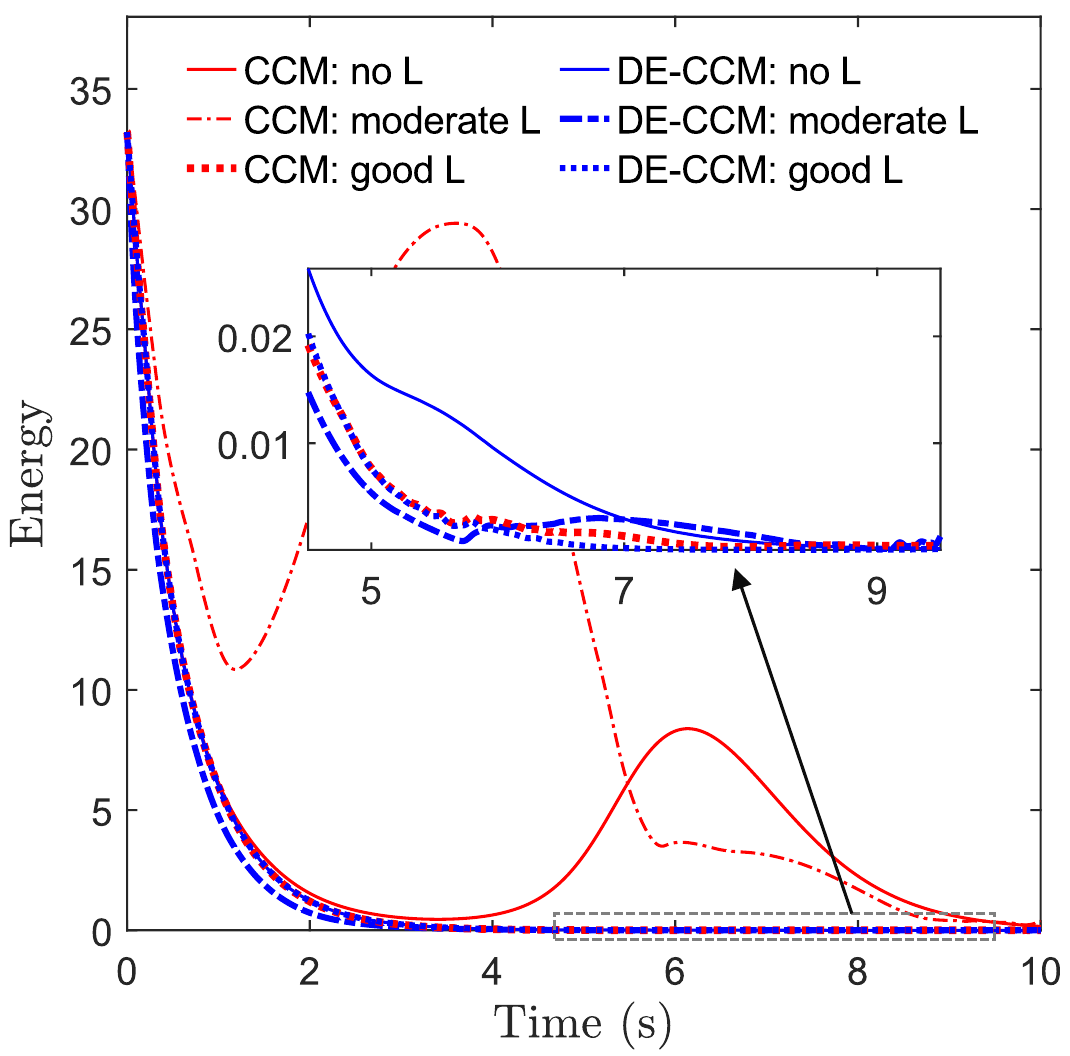}
     \end{subfigure}
     \begin{subfigure}[b]{0.48\textwidth}
         \centering
         \includegraphics[width=\textwidth]{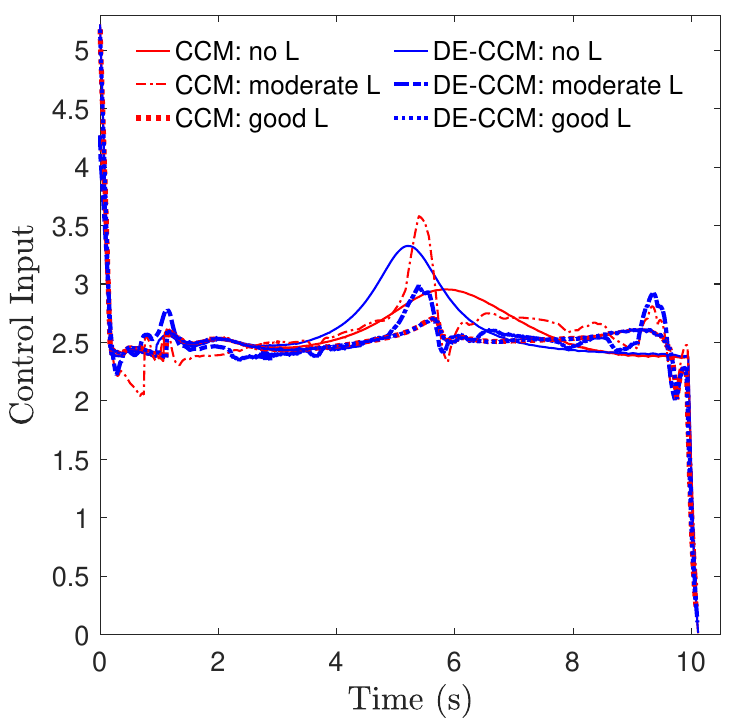}
     \end{subfigure}
     \caption{
    Trajectories of Riemann energy  $E(\xstar,x)$ (left) and the first element of control input (right), i.e., $u_1$, in the presence of no, poor and good learning for Task~1 }
     \label{fig:energy}
\end{figure}
Figure~\ref{fig:energy} shows trajectories of Riemann energy $E(\xstar,x)$ in the presence of no, and good learning for Task~1, while the trajectories for Tasks~2 and 3 are similar and thus omitted. It is evident that $E(\xstar,x)$ under DE-CCM decreased exponentially across all scenarios, irrespective of the quality of the learned mode
{\color{black}
To compare, we implemented three additional controllers, namely, a vanilla CCM controller that disregards the uncertainty or learned model error, adaptive CCM (\textbf{Ad-CCM}) controllers from \cite{lopez2020adaptive-ccm} with true and polynomial regressors, and a robust CCM (\textbf{RCCM}) controller that can be seen as a specific case of a DE-CCM controller with the uncertainty estimate equal to zero and the EEB equal to the disturbance bound $b_d$. 
Since the Ad-CCM controller \cite{lopez2020adaptive-ccm} needs a parametric structure for the uncertainty in the form of $\Phi(x)\theta$ with $\Phi(x)$ being a known regressor matrix, and $\theta$ being the vector of unknown parameters. 
For the no-learning case, we assume that we know the regressor matrix for the original uncertainty $d(x)$. For the learning cases, since we do not know the regressor matrix for the learned model error $\tilde d(x)$, we used a second-order polynomial regressor matrix:
\begin{align}
\Phi(x) = &   \begin{bmatrix}
   1 \\  1
    \end{bmatrix}\begin{bmatrix}
    1 & x & x^2 & z & z^2 & v_x & v_x^2 & v_z & v_z^2 
    \end{bmatrix}.
    \label{eq:ad-ccm-poly-regressors}
\end{align}
Figure~\ref{fig:results-diff-learning} (bottom) shows the tracking control performances yielded by these additional controllers for Task 1 under different learning scenarios. 
The observed trajectories resulting from the CCM controller showed significant deviations from the planned trajectories and occasionally encountered collisions with obstacles, except in the case of good learning. 
Additionally, Ad-CCM yielded poor tracking performance under the moderate learning case. The poor performance could be attributed to the fact that the uncertainty $\tilde d (x)$ may not have a parametric structure, or even if it does, the selected regressor may not be sufficient to represent it. RCCM achieved similar performance as compared to the proposed method, but shows weaker robustness against control input delays, as demonstrated later. 

\subsection{Improved Planning and System Robustness With Learning}

\cref{fig:cost_diff_learning} shows the costs $J$ associated with the actual trajectories achieved by DE-CCM under different learning qualities. As expected, the {\it good} model helped plan better trajectories, leading to reduced costs for all three tasks. 
It is not a surprise that the {\it poor} 
and {\it moderate} models 
led to a temporal increase in the costs for some tasks. In practice, we may not use the poorly learned model 
directly for planning trajectories. DE-CCM ensures that in case one really does so, the planned trajectories can still be tracked well. 
\begin{figure}[h]
    \centering
    \includegraphics[width=0.5\textwidth]{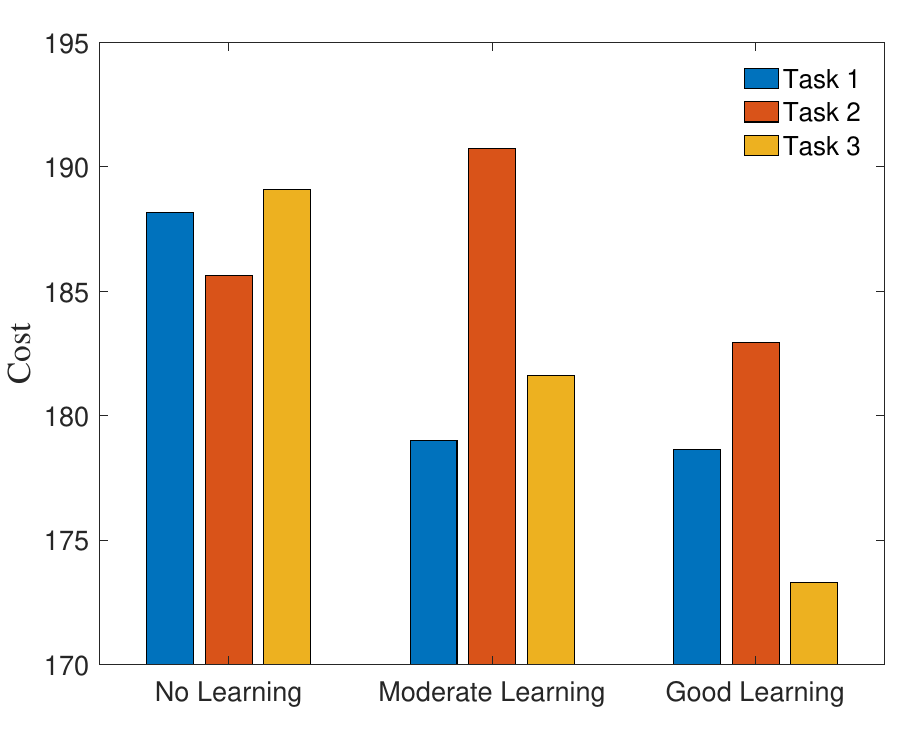}
    \caption{Costs of actual trajectories achieved by DE-CCM  throughout the learning phase}\label{fig:cost_diff_learning}
\end{figure}

\blue{The role of the low-pass filter in protecting system robustness under no learning case is illustrated in \cref{sec:role_filter}}.
We next tested the robustness of RCCM and DE-CCM against input delays under different learning scenarios.
Note that, unlike linear systems for which gain and phase margins are commonly used as robustness criteria, we often evaluate the robustness of nonlinear systems in terms of their tolerance for input delays. Under an input delay of $\Delta t$, the plant dynamics \cref{eq:dynamics} becomes $\dot x = f(x) + B(x)\left (u(t-\Delta t)+d(x)\right)$. 
For these experiments, we leveraged a low-pass filter $\mcC(s) = \frac{30}{s+30}I_2$ to filter the estimated disturbance following \cref{eq:check-d-filter-defn}.  In principle, the EEB of 0.1 used in the previous experiments would not hold anymore as the presence of the filter would lead to a larger error bound according to \cref{eq:estimation_error_bound_filtered}. However, we kept using the same EEB of 0.1 as the theoretical bound according to \cref{eq:estimation_error_bound_filtered} could be quite conservative. The Riemann energy, which indicates the tracking performance for all states, is shown in \cref{fig:delays}. One can see that under both delay cases, DE-CCM achieves smaller and less oscillatory Riemann energy, indicating better robustness and tracking performance, compared to RCCM.
\begin{figure}[h]
    \centering
    \includegraphics[width=0.46\textwidth]{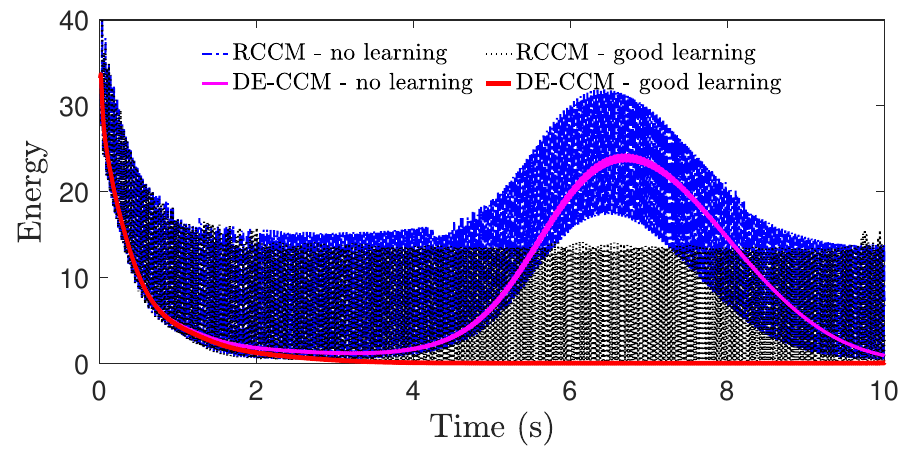}
    \includegraphics[width=0.48\textwidth]{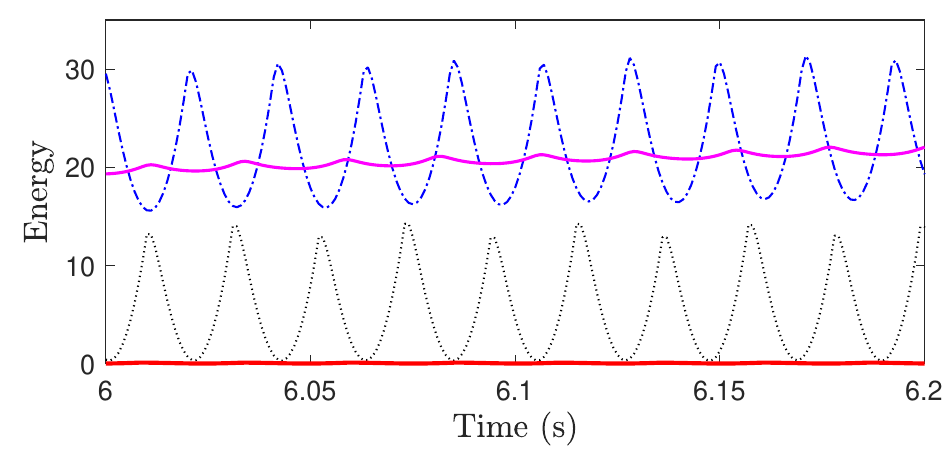}
     \includegraphics[width=0.48\textwidth]{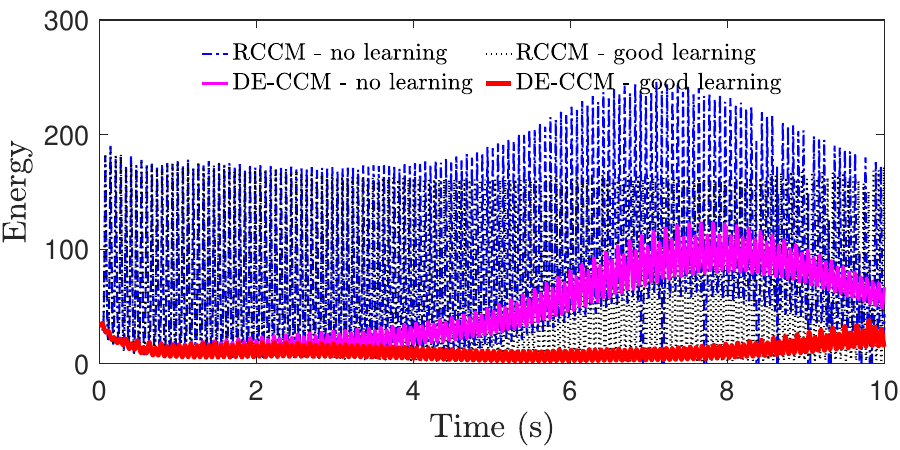}
     \includegraphics[width=0.48\textwidth]{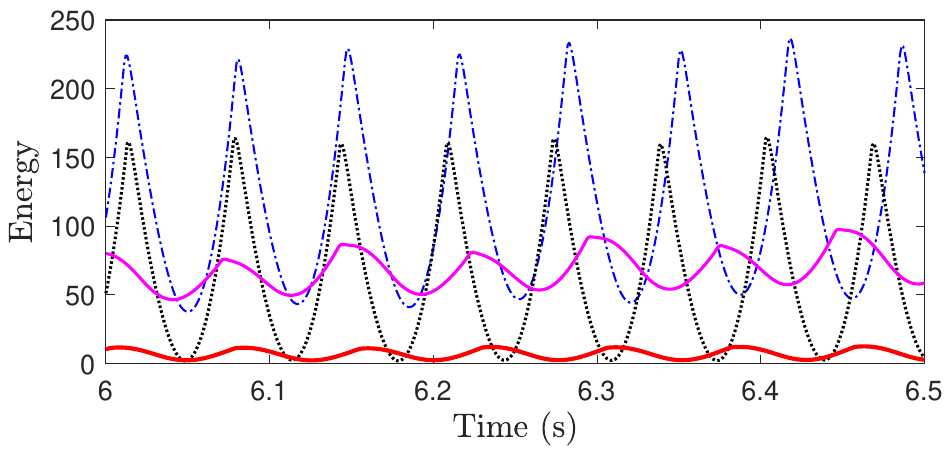}
    \vspace{-2mm}
    \caption{Riemann Energy yielded by DE-CCM and RCCM under different learning cases in with an input delay of 10 ms (top) and 30 ms (bottom). Note that the plots on the right are zoomed-in versions of the corresponding plots on the left. Learning improves the robustness of both DE-CCM and RCCM against input delays. 
    }
    \label{fig:delays}
\end{figure}
Additionally, the robustness of DE-CCM against input delays in the presence of good learning is significantly improved compared to the no-learning case, which illustrates the benefits of incorporating learning. 
This can be explained as follows. The input delay may cause the disturbance estimate $\check d^0(t)$ to be highly oscillatory and a large discrepancy between $\check d^0(t)$ and $d(x(t))$. The low pass filter $\mcC(s)$ can filter out the high-frequency oscillatory component of  $\check d^0(t)$.  
Under good learning, according to \cref{eq:check-d-filter-defn}, the learned model $\hat d(x(t))$ approaches the true uncertainty $d(x(t))$; as a result, the filtered disturbance estimate $\check d(t)$ defined in \cref{eq:check-d-filter-defn} can be much closer to $d(x(t))$ leading to improved robustness and performance, compared to no and moderate learning cases.




\section{Conclusions}\label{sec:conclusion}
This paper presents a disturbance estimation-based contraction control architecture that allows for using model learning tools (e.g., a neural network) to learn uncertain dynamics while guaranteeing exponential trajectory convergence during learning transients under certain conditions. The architecture uses a disturbance estimator to estimate the value of the uncertainty, i.e., the difference between nominal dynamics and actual dynamics, with pre-computable estimation error bounds (EEBs), at each time instant. The learned dynamics, the estimated disturbances, and the EEBs are then incorporated into a robust Riemann energy condition,  which is used to compute the control signal that guarantees exponential convergence to the desired trajectory throughout the learning phase. On the other hand, we show that learning can facilitate better trajectory planning and improve the robustness of the closed-loop system, e.g., against input delays. The proposed framework is validated on a planar quadrotor example.

Future directions could involve addressing broader uncertainties, especially unmatched uncertainties prevalent in practical systems, minimizing the conservatism of the estimation error bound, and demonstrating the efficacy of the proposed control framework with alternative model learning tools.

\section*{Acknowledgments}
This research was funded in part by NASA through the ULI grant \#80NSSC22M0070, and in part by NSF under the RI grant \#2133656.

\bibliographystyle{unsrt}
\bibliography{bib/refs-pan,bib/refs-new,bib/refs-naira}

\appendix
\subsection{Ablation Study on the Role of Low-Pass Filter}\label{sec:role_filter}
We tested the performance of DE-CCM under different filter bandwidths in the presence of a 10 ms input delay. The estimation error bound was fixed to 0.1 for all test cases. The results are shown in Figure~\ref{fig:filters}. It is evident that the low-pass filter helps protect the robustness of the DE-CCM against input delay. Moreover, this protection decreases when the filter bandwidth increases. Note that the no-filter case can be considered equivalent to a filter with infinitely high bandwidth. We also repeated the above tests but without injecting input delay, and did not notice much difference among the different filter cases in terms of both tracking performance and control inputs. 
\begin{figure}[h]
    \centering
    \includegraphics[width=0.48\textwidth]{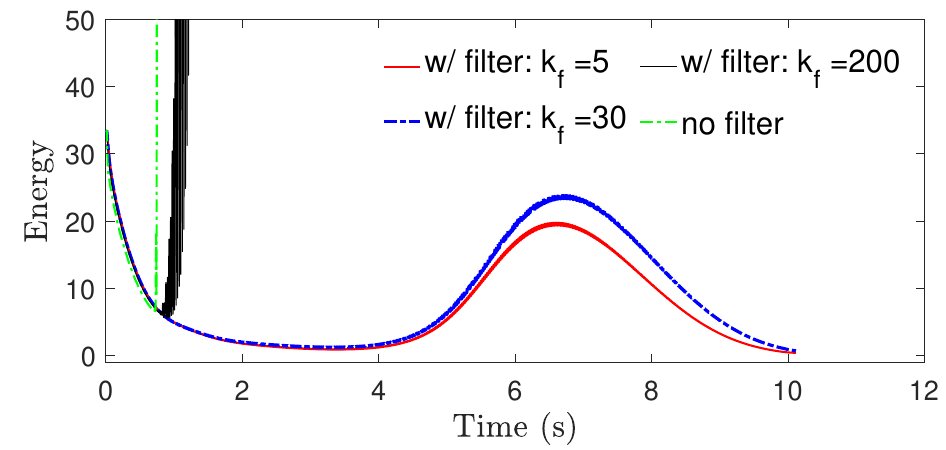}
    \includegraphics[width=0.48\textwidth]{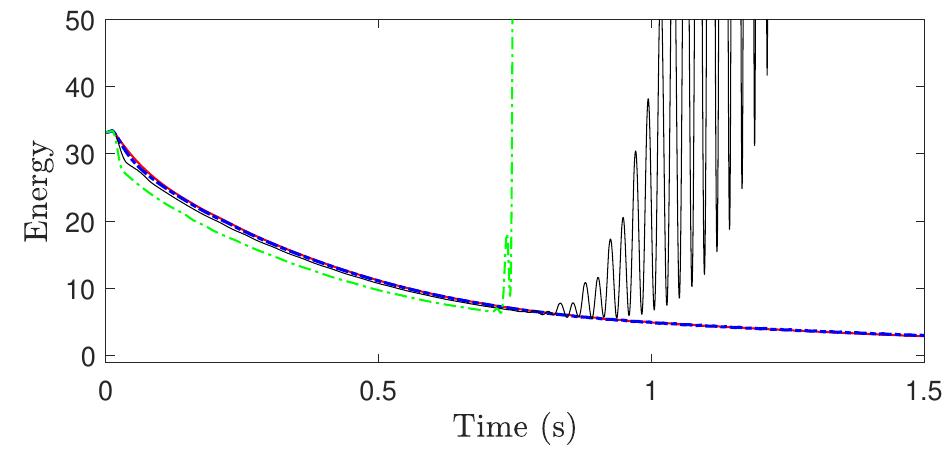}
     \includegraphics[width=0.48\textwidth]{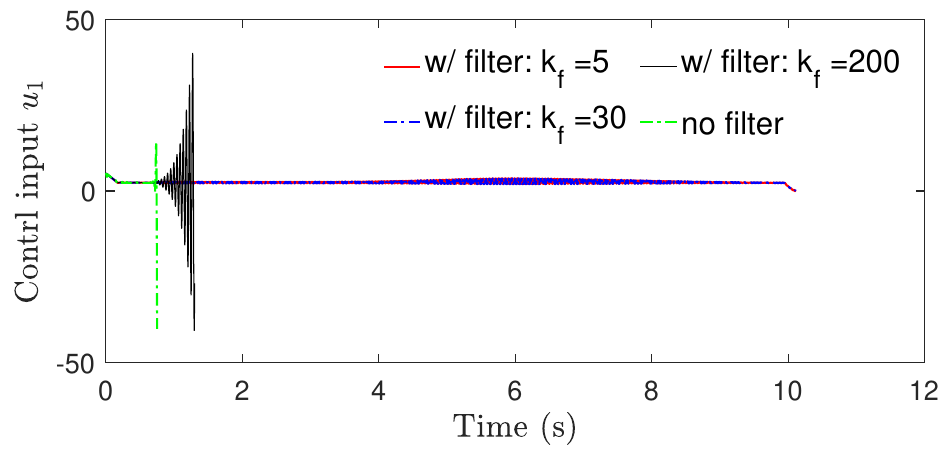}
     \includegraphics[width=0.48\textwidth]{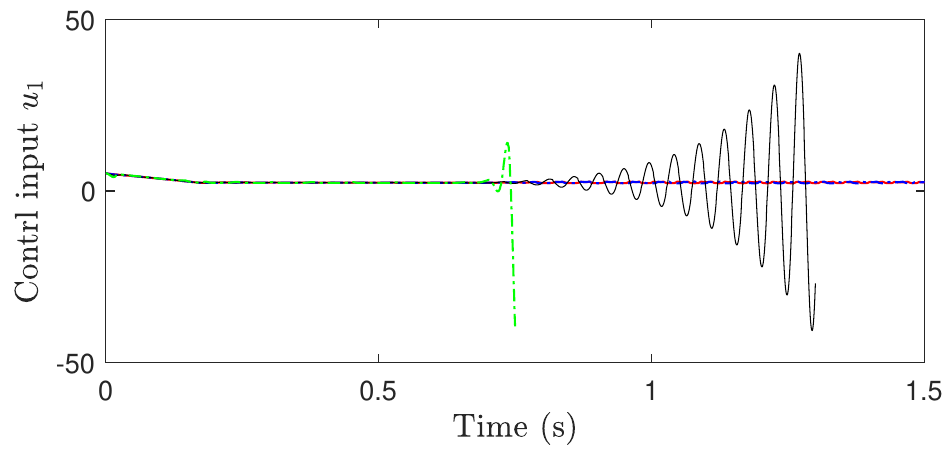}
    \vspace{-2mm}
    \caption{\blue{Riemann Energy (top) and control input (bottom) yielded by DE-CCM under different filters in the presence of 10 ms input delay. The plots on the right are zoomed-in versions of the corresponding plots on the left. $k_f$ denotes the bandwidth of a first-order low-pass filter for both control input channels. Simulations for $k_f=200$ and no-filter cases had to be pre-emptively stopped due to too much deviation of actual states from nominal states, which rendered computation of the geodesic (see Section~\cref{sec:sub-ccm-review}) infeasible.} 
    }
    \label{fig:filters}
\end{figure}
}

\end{document}